\documentclass[12pt,a4paper]{article}

\usepackage{amsfonts}
\usepackage{amsmath}
\usepackage{amssymb}
\usepackage{cite}
\usepackage{graphicx}
\usepackage[svgnames]{xcolor}			
\usepackage{ytableau,tikz}					
\usepackage{amsthm}
\usepackage[colorlinks=true,linkcolor=black,citecolor=black,urlcolor=black]{hyperref}                 		
\usetikzlibrary{intersections, calc,shapes.geometric,patterns,arrows,decorations.pathreplacing} 	

\allowdisplaybreaks[1]
\makeatletter
\@addtoreset{equation}{section}

\makeatother
\def\thefootnote{\ifnum\c@footnote>\z@\textasteriskcentered\@arabic\c@footnote\fi}
\makeatletter
\renewcommand{\footnoterule}{%
\kern-3\p@
\hrule width 0.4\columnwidth
\kern 2.6\p@}
\def\thefootnote{\ifnum\c@footnote>\z@\@arabic\c@footnote\fi}
\makeatother
\makeatletter
\newcommand{\@authornote}[2]{{\def\thefootnote{\fnsymbol{footnote}}\setcounter{footnote}{#1}#2\setcounter{footnote}{0}}}
\newcommand{\authornotemark}[1]{\@authornote#1{\addtocounter{footnote}{-1}\footnotemark}}
\newcommand{\authornotetext}[2]{\@authornote#1{\footnotetext{#2}}}
\makeatother

\begin{document}
\theoremstyle{definition}
\newtheorem*{theorem}{Theorem}
\newtheorem*{proposition}{Proposition}
\newtheorem*{example}{Example}
\newtheorem*{corollary}{Corollary}
\newtheorem*{lemma}{Lemma}
\newtheorem*{definition}{Definition}
\newtheorem*{claim}{Claim}
\newtheorem*{remark}{Remark}
\newtheorem*{observation}{Observation}
\begin{titlepage}

\begin{flushright}
UT--15--04\\
\end{flushright}

\vskip 2.5 cm

\begin{center}
{\Large \bf On the Jeffrey-Kirwan residue of $BCD$-instantons}

\vskip 1.05in

{\Large \bf Satoshi Nakamura${}^\dag$\footnote[0]{${}^\dag${\it E-mail:} \textcolor{black}{satoshi@hep-th.phys.s.u-tokyo.ac.jp}}}
\vskip 0.4in

{\large 
{\it Department of Physics, The University of Tokyo}\\[0.4em]
{\it Hongo 7-3-1, Bunkyo-ku, Tokyo 113-0033, Japan}
}
\vskip 0.1in
\end{center}
\vskip .95in

\begin{abstract}
\noindent We apply the Jeffrey-Kirwan method to compute the multiple integrals for the $BCD$ type Nekrasov partition functions of four dimensional $\mathcal{N}=2$ supersymmetric gauge theories. We construct a graphical distinction rule to determine which poles are surrounded by their integration cycles. We compute the instanton correction of the ``$Sp(0)$" pure super-Yang-Mills theory and find that
\begin{eqnarray*}
Z^{Sp(0)}_{k}=\frac{(-1)^{k}}{2^{k}k!\varepsilon_{1}^{k}\varepsilon_{2}^{k}}
\end{eqnarray*}
for $k\le 8$, which resembles the formula $Z^{U(1)}_{k}=(k!\varepsilon_{1}^{k}\varepsilon_{2}^{k})^{-1}$ for the pure super-Yang-Mills theory with gauge group $U(1)$.
\end{abstract}

\end{titlepage}
\renewcommand{\thefootnote}{\dag\arabic{footnote}}
\setcounter{footnote}{0}

\section{Introduction and Summary} 
From the advent of Seiberg-Witten curves\cite{Seiberg:1994rs}, many works began to appear in order to pursue non-perturbative aspects of supersymmetric gauge theories. In particular, it was found that instanton corrections can be calculable for some theories. 
Nekrasov and others derived instanton partition functions, or Nekrasov partition functions, for four dimensional $\mathcal{N}=2$ gauge theories\cite{Nekrasov:2003af,Nekrasov2003}. They applied a localization method to instanton moduli spaces and gave the partition functions by finite dimensional integrations when the gauge groups are classical. 

When the gauge group is unitary, however, one can avoid direct evaluations of the integral. The $U(1)$ Nekrasov partition function can be rephrased in terms of Hilbert schemes of points on surfaces\cite{nakajimalectures,Nakajima:2003pg} and expressed by a sum over Young diagrams, whose summands can be given by an algebraic manipulation of such diagrams. Each Young diagram corresponds to a fixed point in the instanton moduli space with respect to the localization and  its summand means its weight. Similarly, the $SU(N)$ Nekrasov partition function can be expressed as a sum over fixed points, or $N$-vectors of Young diagrams\cite{Nekrasov:2003af,Flume:2002az,Fucito:2004gi}. In terms of integrations, each fixed point corresponds to a $\mathfrak{S}_{N}$-orbit of poles of the integrand whose sum of residues does not vanish. 

Later, the Alday-Gaiotto-Tachikawa relation\cite{Alday:2009aq} was proposed, which implies that there exists an action of an infinite dimensional nonlinear symmetry, which is of a two dimensional conformal field theory, on (the fixed points of) the instanton moduli space in four dimension. This relation has been proved for some $SU(N)$ gauge theories\cite{Fateev:2009aw,Alba:2010qc,2010arXiv1003.1049Y,schiffmann2013cherednik,maulik2012quantum,Kanno:2013aha,Morozov2014,Matsuo:2014rba}. The fixed points mentioned above have played an role of a good basis where such actions can be seen easily, like simultaneous eigenstates of a Cartan subalgebra of a simple Lie algebra, and been used to construct a certain coherent vector of the nonlinear symmetry.

Taking these results for the unitary gauge group, it may be useful to identify the set of all the fixed points in an instanton moduli space to see how the AGT relation works when the gauge group is general\footnote{The AGT relations for 1-instantons were considered in \cite{Keller:2011ek}. In this case, one can determine all the weights of fixed points without using integral expressions.}. However, the techniques used in the unitary case do not work well for the other gauge groups\footnote{In particular, for the $SO(N)$ or $Sp(N)$ cases, some analogies of the stability condition in \cite{nakajimalectures} for the unitary cases seems to be needed but they are not known at present.}. To the best of our knowledge, one has to evaluate their integrals directly when one identifies the set of fixes points for the $SO(N)$ or $Sp(N)$ cases.
Much efforts have been done for these cases\cite{Nekrasov:2004vw,Fucito:2004gi,Marino:2004cn,Hollands:2010xa}, but were limited for small instanton numbers. 

This limitation may stem not only from their troublesome calculations of multiple integrations but also from a subtle issue of multiple integrations -- the contour choice. For the Nekrasov partition function with one integral variable, the contour lies in the real axis and surrounds the upper half-plane. This expression fixes the integration cycle and then one can perform the integration using Cauchy's formula. Similarly for the partition function with multiple variables, each variable is real and tends to surround the upper half-plane. However, this expression is too naive to identify an integration cycle because one may treat a pole whose (iterated) residue depends on how to surround it\footnote{For example, we have different results from two different cycles
\begin{eqnarray*}
\frac{1}{x(y-x)}\underset{\mathrm{Res}_{x=0}}{\longmapsto}\frac{1}{y}\underset{\mathrm{Res}_{y=0}}{\longmapsto}1,\\
\frac{1}{x(y-x)}\underset{\mathrm{Res}_{y=0}}{\longmapsto}0\underset{\mathrm{Res}_{x=0}}{\longmapsto}0.
\end{eqnarray*}
These results depends on whether the contour that the variable $x$ goes thorough contains the one for the variable $y$. We should split the two circles and determine an integration cycle.}. In other words, one should be careful how to expand an integrant of several variables to get its Laurent series. Also, there may be an pole whose component is real, which means that it is unclear whether the pole is ``in" or ``out" the integration cycle. In the author's previous paper\cite{Nakamura:2014nha}, such unclear poles as well as awkward cancellations of residues appeared after regarding the multiple integral as an iterated integral. These difficulties have motivated the author to construct a better algorithm to compute the integrals representing the Nekrasov partition functions.

Then, when one tackles with a multiple contour integral, one has to specify its integration cycle more clearly. For our cases, the cycle must be determined physically. In recent years, the authors of \cite{Benini:2013xpa,Hori:2014tda,Hwang:2014uwa} considered the choice and showed that
the Jeffrey-Kirwan(JK) residue operator\cite{1993alg.geom..7001J} gives the true cycles in the multiple integration expressions for the Nekrasov partition functions.

Therefore, in this paper, we consider how this JK-residue method works for the $BCD$ type Nekrasov partition functions. We construct a graphical algorithm\footnote{The graphical algorithm concerns the ``multiplicities" of poles in \cite{Hollands:2010xa}. The notion of the multiplicity may arise from the choice of the integration cycle. The pole distribution of the integrand is symmetric under a Weyl group, but the integration cycle may break this symmetry. Then, for a given Weyl orbit of some poles, it is unclear how many poles in the orbit contribute to the partition function. Roughly speaking, the multiplicity for the orbit is the number of its poles in the cycle. The graphical algorithm determines such numbers.} to determine which poles are in the integration cycle and calculate the partition function for the ``$Sp(0)$" case partly. As a result, we get a resultant list of residues in the appendix \ref{listcalc} and observe
\begin{eqnarray}
Z^{Sp(0)}_{k}=\frac{(-1)^{k}}{2^{k}k!\varepsilon_{1}^{k}\varepsilon_{2}^{k}}\quad (k\le 8)
\end{eqnarray}
for the $Sp(0)$ case, which coincides with the Taylor expansion of $\exp (-q/2\varepsilon_{1}\varepsilon_{2})$ at $q=0$. This resembles the one $Z_{k}^{U(1)}=(k!\varepsilon_{1}^{k}\varepsilon_{2}^{k})^{-1}$ for the $U(1)$ case. In the latter $U(1)$ case, it can be rephrased as a combinational formula of Jack polynomials, whose algebraic properties give clearer understandings of the partition function and leaded to the AGT proofs for some unitary cases. Then our observation for the $Sp(0)$ case implies some algebraic backgrounds of its partition function and then may become a key to prove the AGT relations for the $BCD$ cases through a construction of a representation of an infinite dimensional algebra acting on some orthogonal polynomials.

We organize this paper as follows. In section \ref{review}, we first review the Nekrasov partition functions and the JK-residue operator and then consider a sum of residues over a Weyl orbit of poles, representing a fixed point in the localization. We will define a factor $A_{\sigma,\eta}\in\mathbb{N}$, which roughly expresses the number of poles in the integration cycle and then its determination is important to complete the integration. In the appendix \ref{multiplicityfactors}, we comment on relations between the factor and the notion of the multiplicity in \cite{Hollands:2010xa}. In section \ref{section3}, we give an graphical algorithm to determine the $A_{\sigma,\eta}$. This algorithm is applicable for the $BCD$ cases. In section \ref{box_arrangement}, we give a box description of a Weyl orbit of poles\footnote{This box description is not new. It appeared in \cite{Marino:2004cn,Hollands:2010xa}.} and then we apply the above results to the ``$Sp(0)$" instanton partition function. We conclude in section \ref{discuss} with some discussions, giving some related materials that we do not treat in this paper.

\section{Nekrasov partition function and Jeffrey-Kirwan\\ residue}\label{review}
In this section, first we recall the integral expressions of the Nekrasov partition functions and pick up some aspects of the Jeffrey-Kirwan residue. Then we rewrite the Nekrasov partition functions into the summation over all the Weyl orbits of poles in their integrands. We will use the final form (\ref{Zeta}) in the following sections.
\subsection{Nekrasov partition function}
The non-perturbative correction of the partition function for an $\mathcal{N}=2$ supersymmetric gauge theory on the $\Omega$-deformed $\mathbb{R}^{4}$, called the Nekrasov partition function, can be represented as a multiple integral when the gauge group $G$ of the theory is classical\cite{Nekrasov:2003af,Nekrasov:2004vw}.  Here we trace how it was derived and explain some backgrounds why we compute the partition function by another way. {\it One may skip the following explanation if one assumes that the integration cycle that appears in the Nekrasov partition function is the one that the Jeffrey-Kirwan residue operator gives. }

One can lift a four dimensional $\mathcal{N}=2$ supersymmetric gauge theory on the spacetime $\mathbb{R}^{4}$ to a five dimensional $\mathcal{N}=1$ supersymmetric gauge theory on $\mathbb{R}^{4}\times S^{1}_{\beta}$, where the $\beta$ means the circumference of the circle $S^{1}$. The $\Omega$-background is translated as the boundary condition $(z_{1},z_{2},y+\beta)\sim(e^{i\varepsilon_{1}}z_{1},e^{i\varepsilon_{2}}z_{2},y)\ ((z_{1},z_{2})\in\mathbb{C}^{2}\simeq\mathbb{R}^{4}, y\in S^{1})$.

In the weak 5d gauge coupling limit, or $\beta\to 0$, the $k$-instanton sector of the theory reduces to a supersymmetric quantum mechanics whose target space is the finite dimensional $k$-instanton moduli space $\mathcal{M}_{k}$\cite{Nekrasov:1996cz}. Then the $k$-instanton partition function $Z^{\mathrm{(5d)}}_{k}$ for this 5d theory is given by the index
\begin{eqnarray}\label{5dindex}
Z^{\mathrm{(5d)}}_{k}=\mathrm{Tr}_{\mathcal{H}_{k}}(-1)^{F}e^{-\beta H}e^{i\beta(\varepsilon_{1}J_{1}+\varepsilon_{2}J_{2})}e^{i\beta a_{j}\Pi^{j}}e^{i\beta u\cdot K},
\end{eqnarray}
where $\mathcal{H}_{k}$ is the Hilbert space of the quantum mechanics, $J_{1,2}$ are the generators of the Cartan algebra of $SO(4)$, $\Pi^{j}$ is the charge coupling to the vacuum expectation value $a_{j}$ of the adjoint scalar in the vector multiplet, and $K$ is the charge of the other global symmetry in the quantum mechanics, which couples with the parameter $u$. Taking the $\beta\to0$ limit, $Z^{\mathrm{(5d)}}_{k}$ reduces to the Nekrasov partition function $Z_{k}$ for the 4d theory with instanton number $k$.

This index can be written by an integral over the target space as a result of the index theory. Moreover, if the gauge group is classical, one can parametrize the instanton moduli spaces by the ADHM matrices\cite{Atiyah:1978ri}, constrained by the ADHM equation and by the action of the inner symmetry group $\hat{G}$ of instantons. {\it If the target space were a smooth manifold}, one could apply the localization method to this parametrized manifold. Then, one can find all the fixed points with respect to the global symmetry that appears in (\ref{5dindex}) and evaluate their weights around the fixed points. These data allow one to perform the localization explicitly and then one reaches the integral expressions of the instanton partition functions whose integration variables come from the maximal torus of $\hat{G}$. In the 4d limit $\beta\to0$, each integration variable surrounds the upper half-plane. The positive imaginary parts with the coefficients like $\varepsilon_{1,2}$ are imposed for the convergence of the index (\ref{5dindex}). 

In particular, for a pure super-Yang-Mills theory(SYM) with gauge group $G$, the Nekrasov partition function $Z^{G}_{k}$ is written explicitly as follows\cite{Nekrasov:2003af,Nekrasov:2004vw}.
\begin{example}[{\bf the Nekrasov partition functions for pure SYM theories}]
\quad
\vspace{10pt}

\noindent
(i) For $G=U(N)$,
\begin{eqnarray}\label{unitaryN}
Z_{k}^{U(N)}=\frac{1}{k!}\left(\frac{\varepsilon}{2\pi i\varepsilon_{1}\varepsilon_{2}}\right)^{k}\int\frac{d\phi_{1}\cdots d\phi_{k}}{\prod_{j=1}^{k}P(\phi_{j}-\varepsilon_{+})P(\phi_{j}+\varepsilon_{+})}\frac{\Delta(0)\Delta(\varepsilon)}{\Delta(\varepsilon_{1})\Delta(\varepsilon_{2})}
\end{eqnarray}
where 
$\varepsilon=2\varepsilon_{+}=\varepsilon_{1}+\varepsilon_{2}$, 
$P(x)=\prod_{l=1}^{N}(x-a_{l})$ and $\Delta(x)=\prod_{i<j}((\phi_{i}-\phi_{j})^{2}-x^{2})$.
\vspace{10pt}

\noindent
(ii) For $G=SO(N)$ with $N=2n+\chi\ (n\in\mathbb{N}, \chi=0,1)$,
\begin{eqnarray}\label{orthogonalN}
Z_{k}^{SO(N=2n+\chi)}=\frac{(-1)^{k(N+1)}}{2^{k}k!}\left(\frac{\varepsilon}{2\pi i\varepsilon_{1}\varepsilon_{2}}\right)^{k}\int\frac{d\phi_{1}\cdots d\phi_{k}\prod_{j=1}^{k}\phi_{j}^{2}(\phi_{j}^{2}-\varepsilon_{+}^{2})}{\prod_{j=1}^{k}P(\phi_{j}-\varepsilon_{+})P(\phi_{j}+\varepsilon_{+})}\frac{\Delta(0)\Delta(\varepsilon)}{\Delta(\varepsilon_{1})\Delta(\varepsilon_{2})},
\end{eqnarray}
where
$P(x)=x^{\chi}\prod_{l=1}^{n}(x^{2}-a^{2}_{l})$ and $\Delta(x)=\prod_{i<j}((\phi_{i}-\phi_{j})^{2}-x^{2})((\phi_{i}+\phi_{j})^{2}-x^{2})$.
\vspace{10pt}

\noindent
(iii) For $G=Sp(N)$ with $k=2n+\chi\ (n\in\mathbb{N}, \chi=0,1)$,
\begin{align}\label{symplecticN}
Z_{k}^{Sp(N)}=\frac{1}{2}\frac{(-1)^{n}}{2^{n-1}n!}&\left(\frac{\varepsilon}{2\pi i\varepsilon_{1}\varepsilon_{2}}\right)^{n}\left(-\frac{1}{2\varepsilon_{1}\varepsilon_{2}P(\varepsilon_{+})}\right)^{\chi}\nonumber\\
&\times\int\frac{d\phi_{1}\cdots d\phi_{n}}{\prod_{j=1}^{n}P(\phi_{j}-\varepsilon_{+})P(\phi_{j}+\varepsilon_{+})(4\phi_{j}^{2}-\varepsilon_{1}^{2})(4\phi_{j}^{2}-\varepsilon_{2}^{2})}\frac{\Delta(0)\Delta(\varepsilon)}{\Delta(\varepsilon_{1})\Delta(\varepsilon_{2})}
\end{align}
where
$P(x)=\prod_{l=1}^{N}(x^{2}-a_{l}^{2})$ and $\Delta(x)=\prod_{i<j}((\phi_{i}-\phi_{j})^{2}-x^{2})((\phi_{i}+\phi_{j})^{2}-x^{2})\prod_{j=1}^{n}(\phi_{j}^{2}-x^{2})^{\chi}$.
\hspace{\fill}$\blacksquare$
\end{example}

However, this expression is inconvenient to seek a combinational formula for the Nekrasov partition function, which is already known for the case $G=U(N)$\cite{Nekrasov:2003af,Flume:2002az,Fucito:2004gi}. To get a combinational formula, one will need to specify which poles are in the integration cycle. Each fixed point appearing in the localization of the moduli space corresponds to a class of such poles. For $G=SO(N),Sp(N)$, however, there are poles whose components are real and it is unclear whether such poles are in or out the cycles\footnote{For example, there is a pole at $(\phi_{*1},\phi_{*2})=(\varepsilon_{1},0)$ when one considers the 4-instanton correction for the case $G=Sp(N)$.}.  It is desirable to get such a combinational expression not only because we can compute the partition function easily, but also because it was found to be useful to understand the action of an infinite dimensional symmetry to the instanton moduli space, which results in the proof of the AGT relation\cite{Alday:2009aq} for some theories with unitary gauge group\cite{Fateev:2009aw,Alba:2010qc,2010arXiv1003.1049Y,schiffmann2013cherednik,maulik2012quantum,Kanno:2013aha,Morozov2014,Matsuo:2014rba}.

This obstruction to get the sum of residues may stem from the fact that the instanton moduli spaces are generally singular. There are UV singularities where two small instantons approach each other and IR singularities where an instanton goes far away. The latter singularities are resolved by the $\Omega$-background and then we have to resolve the UV singularities. In fact, the known combinational formula for the case with unitary gauge group can be achieved by resolving these singularities\cite{nakajimalectures,Nakajima:2003pg}.

One way to resolve the UV singularities is to lift the quantum mechanics to the gauged linear sigma model (GLSM) with gauge group $\hat{G}$, the internal symmetry of instantons, by adding a vector multiplet to the original matter multiplets. In the limit when the coupling constant $g_{QM}$ of the GLSM becomes strong, $g_{QM}\to\infty$, the original quantum mechanics is restored. What is important is that we can compute an index 
\begin{eqnarray}\label{GLSMindex}
\tilde{Z}_{k}=\mathrm{Tr}_{\tilde{\mathcal{H}}_{k}}(-1)^{F}e^{-\beta H}e^{i\beta(\varepsilon_{1}J_{1}+\varepsilon_{2}J_{2})}e^{i\beta a_{j}\Pi^{j}}e^{i\beta u\cdot K}
\end{eqnarray}
which is parallel to the $Z^{(\mathrm{5d})}_{k}$, where $\tilde{\mathcal{H}}_{k}$ is the Hilbert space of this GLSM. We regard the theory as a circular reduction of a 2d GLSM on a torus $T^{2}$ and connect the index (\ref{GLSMindex}) to the elliptic genus\cite{Benini:2013xpa}, which is a 2d index given by an integration of a meromorphic top form whose integration cycle is explicitly determined. Note that the integration variables $\{\phi_{j}\}_{j=1}^{\mathrm{rank}\hat{G}}$ correspond to the additional vector multiplet by $\phi_{j}=\varphi_{j}+iA_{j}$, where $\varphi_{j}$ and $A_{j}$ are the scalar and the gauge field in the multiplet, which represent a zero mode in $g_{QM}\to 0$. 
\begin{remark}[\!\cite{Hwang:2014uwa}]
One should be careful about a certain continuum spectrum in $\tilde{\mathcal{H}}_{k}$ coming from the additional vector multiplet, which is formed on the $\mathcal{H}_{k}$. Such a continuum can contribute to $\tilde{Z}_{k}$ and prevents one from using the index theorem, but the former contribution should be neglected to obtain the true instanton correction $Z^{(\mathrm{5d})}_{k}$. This continuum stems from that the range of the zero modes $\varphi_{j}+iA_{j}$ in 1d becomes a non-compact cylinder. Its non-compactness may give a continuum spectrum on the original space $\mathcal{H}_{k}$. Thus we should get rid of the contribution coming from the region $|\phi_{j}|\to\infty$ and then connect (\ref{GLSMindex}) to the elliptic genus, which becomes an insertion of a cutoff to the integration. Then if the integrand vanishes well at infinity, or the 5d theory has sufficiently few matters, there is no extra contribution caused by the non-compactness.
\end{remark}
In fact, the authors of \cite{Hwang:2014uwa} performed the above program. For the 4d Nekrasov partition functions, this work plays a role of rewriting the naive integrals into the ones whose integrands are the same ones given by a naive application of the localization method and whose integration cycles are expressed by clearer forms. More precisely, it was found that performing the integrals with respect to these cycles is equivalent to acting the Jeffrey-Kirwan residue operator\cite{1993alg.geom..7001J} to their integrands. 

As a result, we now have a UV completed version of the Nekrasov partition function by use of the GLSMs, and
then we can rewrite them into clearer forms
\begin{eqnarray}
\begin{matrix}
Z=\int d\phi_{1}\cdots d\phi_{n} \mathcal{Z} &\Rightarrow& \tilde{Z}=\sum_{\phi_{*}}\mathrm{JK}\text{-}\mathrm{Res}_{\eta,\phi_{*}}\mathcal{Z}
\end{matrix}
\end{eqnarray}
where the sum is taken over all the poles $\phi_{*}$ of the integrand $\mathcal{Z}$ and the operator $\mathrm{JK}\text{-}\mathrm{Res}_{\eta,\phi_{*}}$ is the Jeffrey-Kirwan residue with respect to a vector $\eta\in\mathbb{R}^{n}$ around the pole $\phi_{*}$. In the following, we explain a detail of the operator. 

\subsection{Jeffrey-Kirwan residue}
The Jeffrey-Kirwan(JK) residue is an operator which maps a rational function to a value, the ``residue" of the function at a point.  Here, following \cite{1999math......3178B,2004InMat.158..453S}, we introduce some minimum essentials of the operator in order to apply it to Nekrasov partition functions\footnote{
In this paper, however, we do not explain why the JK-residue gives the correct integration cycle. See \cite{Benini:2013xpa,Hori:2014tda,Hwang:2014uwa}.}.
\subsubsection{Notations and the definition of the Jeffrey-Kirwan residue}
Here we introduce some notations and then define the JK-residue of rational functions of $n$ variables.

Let ${\bf Q}_{*}\subset\mathbb{R}^{n}$ be a finite subset of nonzero elements and fix a point $\phi_{*}=(\phi_{*j})_{1\le j\le n}\in\mathbb{C}^{n}$. We denote by $\hat{F}_{\phi_{*}}$ the set of formal power series of $n$ variables $\phi_{j}-\phi_{*j}\ (1\le j \le n)$ and by
\begin{eqnarray}
\hat{R}_{{\bf Q}_{*},\phi_{*}}:=\mathrm{span}_{\mathbb{C}}\left\{\frac{f}{\prod_{Q\in{\bf Q}_{*}}\left(Q\cdot (\phi-\phi_{*})\right)^{m_{Q}}} \middle|\ f\in\hat{F}_{\phi_{*}}, m_{Q}\in\mathbb{N}\ (Q\in{\bf Q}_{*}) \right\}
\end{eqnarray}
the ring generated over $\hat{F}_{\phi_{*}}$ by inverting the linear functions $Q\cdot (\phi-\phi_{*})\ (Q\in{\bf Q}_{*})$. This is graded by the degree at $\phi_{*}$ and then we denote its degree $-n$ part by $\hat{R}_{{\bf Q}_{*},\phi_{*}}[-n]$.

Let $\kappa$ be a subset of ${\bf Q}_{*}$. The subset $\kappa$ is called a basis of ${\bf Q}_{*}$ if $\kappa$ forms a basis of $\mathbb{R}^{n}$. For a basis $\sigma$ of ${\bf Q}_{*}$, set
\begin{eqnarray}
f_{\sigma}:=\frac{1}{\prod_{Q\in\sigma}Q\cdot (\phi-\phi_{*})}.
\end{eqnarray}
We call such a fraction $f_{\sigma}$ is basic. Then we denote by $S_{{\bf Q}_{*},\phi_{*}}$ the linear span of the $f_{\sigma}$ where $\sigma$ ranges over all the bases of ${\bf Q}_{*}$. Clearly, $S_{{\bf Q}_{*},\phi_{*}}\subset\hat{R}_{{\bf Q}_{*},\phi_{*}}[-n]$.

Also, we denote by $NS_{{\bf Q}_{*},\phi_{*}}$ the linear span of degree $-n$ rational functions in a form
\begin{eqnarray}
\frac{\psi}{\prod_{Q\in\kappa}\left(Q\cdot (\phi-\phi_{*})\right)^{n_{Q}}}
\end{eqnarray}
where $\psi$ is a polynomial function, $\kappa$ does not contain any bases of ${\bf Q}_{*}$ and the $n_{Q}$ are non-negative integers.

It is proved that one can decompose $\hat{R}_{{\bf Q}_{*},\phi_{*}}[-n]$ into these two spaces:
\begin{proposition}[\!\cite{1999math......3178B}]
We have a direct sum decomposition
\begin{eqnarray}\label{decomp}
\hat{R}_{{\bf Q}_{*},\phi_{*}}[-n]=S_{{\bf Q}_{*},\phi_{*}}\oplus NS_{{\bf Q}_{*},\phi_{*}}
\end{eqnarray}
\hspace{\fill}$\blacksquare$
\end{proposition}
From now on, we restrict ourselves to the case when the ${\bf Q}_{*}$ satisfies the following projectivity condition.
\begin{definition}
We call ${\bf Q}_{*}$ is projective if there is a vector $\delta\in\mathbb{R}^{n}$ which has positive inner products with any elements of ${\bf Q}_{*}$.
\hspace{\fill}$\blacksquare$
\end{definition}

We call a vector $\eta\in\mathbb{R}^{n}$ generic if $\eta\notin\mathrm{span}_{\mathbb{R}}\{Q_{1},\cdots,Q_{n-1}\}$ for any $n-1$ elements $Q_{j}$\quad$\ (1\le j\le n-1)$ of ${\bf Q}_{*}$. 

Fix a generic vector $\eta$ and then we define the Jeffrey-Kirwan residue of a basic fraction $f_{\sigma}$ for the projective ${\bf Q}_{*}$ by
\begin{eqnarray}\label{JK}
\mathrm{JK}\text{-}\mathrm{Res}_{\eta,\phi_{*}}\left(f_{\sigma}d\phi_{1}\wedge\cdots\wedge d\phi_{n}\right):=\begin{cases}|\mathrm{det}(Q_{i}\cdot e_{j})|^{-1}&(\eta\in\mathrm{Cone}(Q_{1},\cdots,Q_{n}))\\ 0&(\text{otherwise}),
\end{cases}
\end{eqnarray}
where $\{ e_{j}\}_{1\le j\le n}$ is the canonical basis of $\mathbb{R}^{n}$, $\sigma=\{Q_{1},\cdots, Q_{n}\}$, and
\begin{eqnarray}
\mathrm{Cone}(\sigma)=\mathrm{Cone}(Q_{1},\cdots,Q_{n}):=\left\{ \sum_{j=1}^{n}a_{j}Q_{j}\in\mathbb{R}^{n}\middle|\ a_{1},\cdots a_{n}>0\right\}.
\end{eqnarray}
\begin{proposition}[\!\cite{1999math......3178B}]
Given that ${\bf Q}_{*}$ is projective, the definition in (\ref{JK}) gives a well-defined linear functional $\mathrm{JK}\text{-}\mathrm{Res}_{\eta,\phi_{*}}$ on $S_{{\bf Q}_{*},\phi_{*}}$\footnote{We define the JK-residue of a rational function by that of the meromorphic top form obtained by multiplying a fixed top form $d\phi_{1}\wedge\cdots\wedge d\phi_{n}$ to the function.}.
\hspace{\fill}$\blacksquare$
\end{proposition}
\begin{remark}
Note that there are linear relations among the basic fractions and then the above proposition claims that the JK-residue of each element in $S_{{\bf Q}_{*},\phi_{*}}$ does not depend on how we decompose it into some basic fractions.
\end{remark}
\begin{definition}[{\bf The Jeffrey-Kirwan residue}]
For a projective ${\bf Q}_{*}$ and a generic vector $\eta$, the Jeffrey-Kirwan residue
\begin{eqnarray}
\mathrm{JK}\text{-}\mathrm{Res}_{\eta,\phi_{*}}:\hat{R}_{{\bf Q}_{*},\phi_{*}}\to \mathbb{C}
\end{eqnarray}
is defined to be the composite of
\begin{eqnarray}\label{JKJK}
\hat{R}_{{\bf Q}_{*},\phi_{*}}\xrightarrow{\text{Project}}\hat{R}_{{\bf Q}_{*},\phi_{*}}[-n]\xrightarrow{\text{(\ref{decomp})}} S_{{\bf Q}_{*},\phi_{*}}\xrightarrow{\text{(\ref{JK})}}\mathbb{C}
\end{eqnarray}
\hspace{\fill}$\blacksquare$
\end{definition}
\begin{remark}
The union of hyperplanes of non-generic points divides the whole space $\mathbb{R}^{n}$ into connected components of generic points. Each connected component is often called a chamber. Note that two generic points in the same chamber give the same JK-residue. In other words, a chamber defines a JK-residue operator.
\end{remark}
\subsubsection{A setup to operate the JK-residue to the Nekrasov partition function}
Now we want to apply the JK-residue to the Nekrasov partition function $Z$. We denote its integrand by $\mathcal{Z}$, which is a rational function of $n$ variables $\phi_{1},\cdots,\phi_{n}$, and its denominator is factorized into a product of degree-one polynomials. 

Each factor in the denominator describes an affine hyperplane $H$ in $\mathbb{C}^{n}$. We can choose the coefficient vector $Q$, which defines $H=\{\phi\in\mathbb{C}^{n}|Q\cdot \phi+\mathrm{const}=0\}$, is in $\mathbb{R}^{n}$. We fix the vector $Q$ for each hyperplane $H$ and define $\widetilde{\bf Q}$ to be the set of all such vectors. 

\begin{definition}[\bf $\widetilde{\bf Q}$ for the Nekrasov partition functions]
We fix the vector for each factor in the denominators of the pure SYM's Nekrasov partition functions as follows:
\begin{eqnarray}\label{Qass}
\begin{array}{ccl}
\pm\phi_{j}\pm a_{l}-\varepsilon_{+}, \pm\phi_{j}-\varepsilon_{1,2,+}, \pm\phi_{j}-\frac{\varepsilon_{1,2}}{2}&\Rightarrow&Q=\pm e_{j}\\
\pm\phi_{i}\pm\phi_{j}-\varepsilon_{1,2}&\Rightarrow&Q=\pm e_{i}\pm e_{j}
\end{array}
\end{eqnarray}
Then we have two sets of vectors
\begin{eqnarray}
\widetilde{\bf Q}^{A}=\{\pm e_{j}|j=1,\cdots,n\}\sqcup\{\pm(e_{i}-e_{j})|1\le i<j\le n\}
\end{eqnarray}
for $G=U(N)$ and 
\begin{eqnarray}
\widetilde{\bf Q}^{BCD}=\{\pm e_{j}|j=1,\cdots,n\}\sqcup\{\pm(e_{i}\pm e_{j})|1\le i<j\le n\}
\end{eqnarray}
for $G=SO(N),Sp(N)$, where $\{ e_{j}| 1\le j\le n\}$ is the canonical basis of $\mathbb{R}^{n}$.\hspace{\fill}$\blacksquare$
\end{definition}
\begin{remark}
The above choices of signs reflect that the physically motivated contour of each $\phi_{j}$ should separate $+\varepsilon_{1,2}$ and $-\varepsilon_{1,2}$.
\end{remark}

Now we take a point $\phi_{*}\in\mathbb{C}^{n}$. Then we denote by ${\bf Q}_{*}={\bf Q}_{*}(\phi_{*})$ the subset of $\widetilde{\bf Q}$ that consists of all the coefficient vectors which come from the hyperplanes that contain $\phi_{*}$. The Taylor expansion of the factors except those in the denominator that vanish at $\phi_{*}$ gives an element of $\hat{R}_{{\bf Q}_{*},\phi_{*}}$.

We can see from (\ref{JKJK}) that the JK-residue at $\phi_{*}$ vanishes if ${\bf Q}_{*}$ does not contain any basis. Then there are only finite points where the JK-residue of $\mathcal{Z}$ do not vanish.

Also, for Nekrasov partition functions, the projectivity condition holds for each $\phi_{*}$. The $n$ components $\phi_{j*}$ of $\phi_{*}$ are partially ordered by its coefficients on $\varepsilon_{1}$ and $\varepsilon_{2}$. Thanks to the minus on every $\varepsilon_{1,2,+}$ in (\ref{Qass}), this order gives a vector $\delta$ such that $\delta\cdot Q>0$ for any $Q\in{\bf Q}_{*}$.

Then, provided a generic vector $\eta$ with respect to the $\widetilde{\bf Q}$, we get a finite sum of the JK-residues
\begin{eqnarray}
Z_{\eta}\equiv\sum_{\phi_{*}}\mathrm{JK\text{-}Res}_{\eta,\phi_{*}}\mathcal{Z}.
\end{eqnarray}
by regarding the integrand $\mathcal{Z}$ as an element of $\hat{R}_{{\bf Q}_{*},\phi_{*}}$ for each point $\phi_{*}$.
\begin{remark}
We should choose the $\eta$ physically. In this regard, the authors of \cite{Benini:2013xpa,Hori:2014tda,Hwang:2014uwa} identified the vector as follows,
\begin{eqnarray}
\begin{array}{lcl}
G=U(N)&\Rightarrow& \eta=(1,1,\cdots,1)\\
G=SO(N),Sp(N)&\Rightarrow& \eta\ \text{is arbitrary}.
\end{array} 
\end{eqnarray} 
In the following, for easier calculations, we will choose $\eta=(1,\alpha,\cdots,\alpha^{n-1})\in\mathbb{R}^{n}$ with a large $\alpha$ when we tackle with the $BCD$-type Nekrasov partition functions.
\end{remark}
\subsubsection{The sum of Jeffrey-Kirwan residues over a Weyl orbit}
Before going to explicit calculations, it may be useful to consider a sum of JK-residues over a Weyl orbit $[\phi_{*}]$ of a pole $\phi_{*}$. 

Let $\hat{G}$ be a simple Lie group with rank $n$ and $\hat{\mathfrak{g}}=\mathrm{Lie}(\hat{G})$. Set $(\phi_{1},\cdots,\phi_{n})$ be a linear coordinate of the Cartan subalgebra of the complexified Lie algebra $\hat{\mathfrak{g}}_{\mathbb{C}}$. We denote the Weyl group of $\hat{G}$ by $W\subset O(n)$. Given a Weyl symmetric rational function $\mathcal{Z}=\mathcal{Z}(\phi_{1},\cdots,\phi_{n})$, let us consider the sum of all the JK-residues
\begin{eqnarray}
Z_{\eta}\equiv\sum_{\phi_{*}}\mathrm{JK\text{-}Res}_{\eta,\phi_{*}}\mathcal{Z}\ .
\end{eqnarray}
Here we assume that the the denominator of $\mathcal{Z}$ is completely factorized into distinct affine hyperplanes and the JK-residues are non-zero only at finite $\phi_{*}$s. Also we assume that the ${\bf Q}_{*}$ for each pole $\phi_{*}$ is contained in $\mathbb{R}^{n}$ and is projective. We take the vector $\eta$ to be at a regular point with respect to any ${\bf Q}_{*}$s.

Thanks to the Weyl symmetry of $\mathcal{Z}$, we can relate JK-residue operators in the same Weyl orbit as follows:
\begin{proposition}
For $g\in W$, we have
$
\mathrm{JK\text{-}Res}_{g\cdot\eta,g\cdot\phi_{*}}\mathcal{Z}=\mathrm{JK\text{-}Res}_{\eta,\phi_{*}}\mathcal{Z}
$.
\hspace{\fill}$\blacksquare$
\end{proposition}
\begin{proof}
The action of an element $g\in W$ gives a map from $\hat{R}_{{\bf Q}_{*}(\phi_{*}),\phi_{*}}[-n]$ to $\hat{R}_{{\bf Q}_{*}(g\cdot\phi_{*}),g\cdot\phi_{*}}[-n]$, which preserves their grades. More precisely, for a vector $P\in\mathbb{R}^{n}$ and $g\in W$, we have 
\begin{eqnarray}
P\cdot (g^{-1}\cdot\phi-\phi_{*})=(g\cdot P)\cdot(\phi-g\cdot\phi_{*})\ .
\end{eqnarray}
This action also conserves the decomposition $\hat{R}_{{\bf Q}_{*}}[-n]=S_{{\bf Q}_{*}}\oplus NS_{{\bf Q}_{*}}$ and maps a basic fraction $f_{\sigma}\in S_{{\bf Q}_{*},\phi_{*}}$ to another basic fraction $f_{g\cdot\sigma}\in S_{{\bf Q}_{*},g\cdot\phi_{*}}$, where $g\cdot\sigma:=\{ g\cdot Q\ |\ Q\in\sigma\}$ forms a basis of ${\bf Q}_{*}(g\cdot\phi_{*})$.

Now let us consider the map $\hat{R}_{{\bf Q}_{*},\phi_{*}}\to S_{{\bf Q}_{*},\phi_{*}}$ in the part of (\ref{JKJK}) at $\phi_{*}$. Thanks to the Weyl symmetry of $\mathcal{Z}$, the above discussion shows that if we have a linear combination $\sum_{\sigma} c_{\sigma}(\phi_{*})f_{\sigma}\in S_{{\bf Q}_{*},\phi_{*}}$ of basic fractions as the image of $\mathcal{Z}$ by the map at $\phi_{*}$, then we also have $\sum_{\sigma} c_{\sigma}(\phi_{*})f_{g\cdot\sigma}\in S_{{\bf Q}_{*},g\cdot\phi_{*}}$ as the one at the position $g\cdot\phi_{*}$.

For the projectivity of $Q_{*}$, the JK-residue of the $\sum_{\sigma} c_{\sigma}(\phi_{*})f_{\sigma}$ at $\phi_{*}$ is
\begin{eqnarray}
\mathrm{JK\text{-}Res}_{\eta,\phi_{*}}\left(\sum_{\sigma} c_{\sigma}(\phi_{*})f_{\sigma}\right)=\sum_{\sigma} \frac{c_{\sigma}(\phi_{*})}{|\mathrm{det}(Q^{\sigma}_{i}\cdot e_{j})|}{\bf 1}_{\mathrm{Cone}(\sigma)}(\eta),
\end{eqnarray}
where we write $\sigma=\{ Q^{\sigma}_{1},\cdots,Q^{\sigma}_{n}\}$ and ${\bf 1}_{A}$ is the characteristic function for a subset $A\subset\mathbb{R}^{n}$. On the other hand, we have
\begin{eqnarray}\label{JKoffractions}
\mathrm{JK\text{-}Res}_{g\cdot\eta,g\cdot\phi_{*}}\left(\sum_{\sigma} c_{\sigma}(\phi_{*})f_{g\cdot\sigma}\right)=\sum_{\sigma} \frac{c_{\sigma}(\phi_{*})}{|\mathrm{det}(Q^{g\cdot\sigma}_{i}\cdot e_{j})|}{\bf 1}_{\mathrm{Cone}(g\cdot\sigma)}(g\cdot\eta).
\end{eqnarray}
Clearly, ${\bf 1}_{\mathrm{Cone}(g\cdot\sigma)}(g\cdot\eta)={\bf 1}_{\mathrm{Cone}(\sigma)}(\eta)$. Also we have $|\mathrm{det}(Q^{g\cdot\sigma}_{i}\cdot e_{j})|=|\mathrm{det}(Q^{\sigma}_{i}\cdot e_{j})|$ since $W\subset O(n)$. Then we conclude that $\mathrm{JK\text{-}Res}_{g\cdot\eta,g\cdot\phi_{*}}\mathcal{Z}=\mathrm{JK\text{-}Res}_{\eta,\phi_{*}}\mathcal{Z}$.
\end{proof}

Next we consider a sum of JK-residues over the Weyl orbit of $\phi_{*}$ with respect to a fixed chamber $\eta$. From (\ref{JKoffractions}), we have
\begin{eqnarray}
\frac{1}{C_{[\phi_{*}]}}\sum_{g\in W}\mathrm{JK\text{-}Res}_{\eta,g\cdot\phi_{*}}\mathcal{Z}
=\frac{1}{C_{[\phi_{*}]}}\sum_{\sigma}\frac{c_{\sigma}(\phi_{*})}{|\mathrm{det}(Q^{\sigma}_{i}\cdot e_{j})|}\sum_{g\in W}{\bf 1}_{\mathrm{Cone}(g\cdot\sigma)}(\eta).
\end{eqnarray}
where $\mathcal{Z}\in\hat{R}_{{\bf Q}_{*},\phi_{*}}\mapsto\sum_{\sigma} c_{\sigma}(\phi_{*})f_{\sigma}\in S_{{\bf Q}_{*},\phi_{*}}$ and $C_{[\phi_{*}]}$ is the order of the stabilizer subgroup at $\phi_{*}$. Especially, the $\eta$ only appears in the factor $A_{\sigma,\eta}\equiv\sum_{g\in W}{\bf 1}_{\mathrm{Cone}(g\cdot\sigma)}(\eta)\in\mathbb{N}$.

As a result, we have a following formula\footnote{The sum $\sum_{\sigma\in\mathfrak{B}({\bf Q}_{*})}$ is taken over all the bases of ${\bf Q}_{*}$. In other words the range depends on $\phi_{*}$. So we denote by $\mathfrak{B}({\bf Q}_{*})\subset\widetilde{{\bf Q}}^{BCD}$ the set of all the bases of ${\bf Q}_{*}$ and express the dependence explicitly. }:
\begin{proposition} With the notations above, we have
\begin{eqnarray}\label{Zeta}
Z_{\eta}=\sum_{\phi_{*}}\mathrm{JK\text{-}Res}_{\eta,\phi_{*}}\mathcal{Z}=\sum_{[\phi_{*}]}\frac{1}{C_{[\phi_{*}]}}\sum_{\sigma\in\mathfrak{B}({\bf Q}_{*})}\frac{A_{\sigma,\eta}c_{\sigma}(\phi_{*})}{|\mathrm{det}(Q^{\sigma}_{i}\cdot e_{j})|}.
\end{eqnarray}
\hspace{\fill}$\blacksquare$
\end{proposition}
This formula will be used in the following of this paper and we will concentrate on how to determine the number $A_{\sigma,\eta}$.
\begin{remark}
The $\eta$ only affects the number $A_{\sigma,\eta}$. This structure implicitly appears in the case of the $U(N)$ Nekrasov partition function. In this case, one encounters poles that can be described as an $N$-vector of Young tableaux. The corresponding vector of Young diagrams implies its orbit. When $\hat{G}=U(n)$ and $\eta=(1,1,\cdots,1)\in\mathbb{R}^{n}$, we have 
\begin{eqnarray}\label{Amulti}
A_{\sigma,\eta}=\begin{cases} n!&(\eta\in\mathrm{Cone}(\sigma))\\
0&(\text{otherwise})
\end{cases}.
\end{eqnarray}
In other words, any two poles in the same orbit contribute the result with the same value. So we usually see the famous formula of the $U(N)$ instanton counting whose summation is taken over all the vector of Young diagrams, rather than Young tableaux.\end{remark}

Going back to the $BCD$ Nekrasov partition functions, we want the formula of $A_{\sigma,\eta}$ which is an analogue of (\ref{Amulti}). For the case that one can extract only one basic fraction at the pole $\phi_{*}$, the $A_{\sigma,\eta}$ was partly calculated and called as the multiplicity factor in \cite{Hollands:2010xa}. Unlike the unitary case, the $A_{\sigma,\eta}$ varies for each orbit. In the appendix \ref{multiplicityfactors}, we comment more on the multiplicity from a viewpoint of the notion $A_{\sigma,\eta}$. 

Thanks to previous works, especially to \cite{Hwang:2014uwa}, we know we take the $\eta$ arbitrarily in these $BCD$ cases. In the next section, we will fix the $\eta$ and give an algorithm that computes $A_{\sigma,\eta}$ with respect to this fixed vector.

\begin{remark}
For the arbitrariness of a choice of $\eta$, we may take an average of $Z_{\eta}$ over all the chambers. Set $N_{\widetilde{{\bf Q}}}$ to be the number of chambers for $\widetilde{{\bf Q}}$ and we have
\begin{align}
\frac{1}{N_{\widetilde{{\bf Q}}}}\sum_{\eta}Z_{\eta}&=\sum_{[\phi_{*}]}\frac{1}{C_{[\phi_{*}]}}\sum_{\sigma\in\mathfrak{B}({\bf Q}_{*})}\frac{c_{\sigma}(\phi_{*})}{|\mathrm{det}(Q^{\sigma}_{i}\cdot e_{j})|}\sum_{\eta}\frac{A_{\sigma,\eta}}{N_{\widetilde{{\bf Q}}}}\nonumber\\
&=\sum_{[\phi_{*}]}\frac{|W|}{C_{[\phi_{*}]}}\sum_{\sigma\in\mathfrak{B}({\bf Q}_{*})}\frac{c_{\sigma}(\phi_{*})}{|\mathrm{det}(Q^{\sigma}_{i}\cdot e_{j})|}\mathrm{Prop}(\mathrm{Cone}(\sigma))
\end{align}
where $\mathrm{Prop}(\mathrm{Cone}(\sigma))$ is the proportion of the cone in $\mathbb{R}^{n}$, which is no other than the ratio of the solid angle of the cone, or of the number of the chambers in the cone.
\end{remark}

\section{Distinction rule}\label{section3}
In the previous section, we have encountered the factor
\begin{eqnarray}
A_{\sigma,\eta}\equiv\sum_{g\in W}{\bf 1}_{\mathrm{Cone}(g\cdot\sigma)}(\eta)\in\mathbb{N}.
\end{eqnarray}
So it is important to judge whether $\eta\in\mathrm{Cone}(\sigma)$ or not. In principle, one can do it by solving a system of linear equations. In this section, however, we give another way. We take a special vector as $\eta$, and then provide the main theorem of this section that claims a distinction rule to see whether $\eta\in\mathrm{Cone}(\sigma)$ or not for a given cone of a basis $\sigma$ of $\widetilde{{\bf Q}}^{BCD}$. The rule will be expressed graphically, and we will see in the next section that this graphical notation matches the Weyl group actions on poles.

In this section, we firstly introduce some notations for clear discussions. Then we give the distinction rule.
\subsection{Notations}
Here we introduce some notations to see discussions in this paper more clearly. 
Firstly, we introduce a graphical notation for the elements of $\widetilde{{\bf Q}}^{BCD}$. 

For each element of $\widetilde{{\bf Q}}^{BCD}$, we write it as follows;
\begin{eqnarray}
\begin{matrix}
\pm e_{i}& \Rightarrow&
\begin{tikzpicture}[auto,node distance=1cm,
  thick,main node/.style={circle,draw,minimum size=3mm},baseline=-.1cm]
  \node[main node] (1) {$$};
  \node[above] at (1.north){$i$};
  \draw (1) to [in=150,out=210,loop] node {$\pm$} (1);
\end{tikzpicture}
\\
+e_{i}\pm e_{j} & \Rightarrow&
\begin{tikzpicture}[auto,node distance=1cm,
  thick,main node/.style={circle,draw,minimum size=3mm},baseline=-.1cm]
  \node[main node] (1) {$$};
  \node[above] at (1.north){$i$};
  \node[main node] (2) [right of=1] {$$};
  \node[above] at (2.north){$j$};
  \draw (1) to node[very near start,below] {$+$} node[very near end,below]{$\pm$}(2);
\end{tikzpicture}\ .
\end{matrix}
\end{eqnarray}
Following this notation, we define the graph of a subset ${\bf Q}_{*}\subset\widetilde{{\bf Q}}^{BCD}$ which consists of $n$ vertices with numbers and $|{\bf Q}_{*}|$ lines with signs on their ends\footnote{This graph may be disconnected and have some isolated vertices.}. 
\begin{example} The graph of ${\bf Q}_{*}=\{ e_{2}, e_{1}-e_{2}\}$ is
\begin{tikzpicture}[auto,node distance=1cm,
  thick,main node/.style={circle,draw,minimum size=3mm},baseline=-.1cm]
  \node[main node] (1) {$$};
  \node[above] at (1.north){$2$};
  \node[main node] (2) [right of=1] {$$};
  \node[above] at (2.north){$1$};
  \draw (1) to [in=150,out=210,loop] node[] {$+$}  node[near end, above]{$$}(1);
  \draw (1) to node[very near start,below] {$-$} node[above]{$$} node[very near end,below]{$+$}(2);
  \end{tikzpicture}\ .\hspace{\fill}$\blacksquare$
\end{example}

In the following of this paper, we only treat such a graph and then we call it as a graph for short. We often abbreviate numbers or signs in a graph. Also, if $u$ is a vertex in a graph, we write by $m(u)$ the number assigned to $u$.

\begin{remark}Note that the Weyl group $\mathfrak{S}_{n}\ltimes\mathbb{Z}^{n}_{2}$ naturally acts on a graph. More precisely, $\mathfrak{S}_{n}$ permutates the assigned numbers for the vertices and $\mathbb{Z}^{n}_{2}$ flips the signs on the lines.\end{remark}

We also introduce some notions for tree graphs, a connected graphs with no cycles.  Note that we fix a vertex in a tree graph and call it the root of the tree.

For given a vertex $u$ in a tree graph, we call that a vertex $v$ is descendant to the vertex $u$ when the path that connects $v$ and the root has the vertex $u$, and define $D(u)$ as the set of all the descendants of $u$. (Here we declare that $u$ is an element of $D(u)$.) Also, we define $M(u)\equiv\max_{v\in D(u)}m(v)$.

Finally, we introduce a concept of $\eta$-orientation, only for abbreviations. 
\begin{definition}
For given a subset $\{Q_{1},\cdots,Q_{l}\}$ of\ $\widetilde{{\bf Q}}^{BCD}$, we say that the set and its graph are $\eta$-oriented if the set satisfies $\eta\in\mathrm{Cone}(Q_{1},\cdots,Q_{l})$. \hspace{\fill}$\blacksquare$\end{definition}
\subsection{The specific choice of a chamber}
In this paper, we choose 
\begin{eqnarray}
\eta=(1,\alpha,\alpha^{2},\cdots,\alpha^{n-1})\in\mathbb{R}^{n},
\end{eqnarray}
where $\alpha>0$ is a sufficient large parameter. Before we consider JK-residues with respect to this vector, we should show the following lemma.
\begin{lemma}
Any set of $q(<n)$ vectors in $\widetilde{{\bf Q}}^{BCD}$ is not $\eta$-oriented. In other words, $\eta$ is generic.\hspace{\fill}$\blacksquare$
\end{lemma}
In the following, we prove this lemma. This is a corollary of the following lemma.
\begin{lemma}
For any sequences $(p_{j})_{1\le j\le n}$ of integers such that $p_{1}<p_{2}<\cdots<p_{n}$, the vector $\eta_{(p_{1},\cdots,p_{n})}\equiv(\alpha^{p_{1}},\cdots,\alpha^{p_{n}})\in\mathbb{R}^{n}$ can not be expressed as a linear combination of arbitrary $n-1$ vectors in $\widetilde{{\bf Q}}^{BCD}$.\hspace{\fill}$\blacksquare$
\end{lemma}
\begin{proof}
We prove it by the mathematical induction for $n$. The $n=1$ case holds because $\eta\neq 0$.
Now we assume this claim holds for the $n<l$ cases and consider the $n=l$ case. If the claim does not hold for this case, we have a sequence of integers $p_{1}<\cdots<p_{n}$, vectors $Q_{1},\cdots,Q_{n-1}\in\widetilde{{\bf Q}}^{BCD}$ and $a_{1},\cdots, a_{n-1}\in\mathbb{R}$ such that
\begin{eqnarray}\label{eq1}
a_{1}Q_{1}+\cdots+a_{n-1}Q_{n-1}=\eta_{(p_{1},\cdots,p_{n})}.
\end{eqnarray}

Consider a projection $Q\in\widetilde{{\bf Q}}^{BCD}\mapsto Q'=Q-(e_{1},Q)e_{1}\in\mathbb{R}^{n-1}$. This projection gives a map $\widetilde{{\bf Q}}^{BCD}$ to the one dimensional lower $\widetilde{{\bf Q}}^{BCD}\sqcup\{0\}\subset\mathbb{R}^{n-1}$. So if $Q_{1}=e_{1}$, this projection gives
\begin{eqnarray}
a_{2}Q'_{2}+\cdots+a_{n-1}Q'_{n-1}=\eta_{(p_{2},p_{3},\cdots,p_{n})},
\end{eqnarray}
which contradicts the assumption. Therefore $Q_{1},\cdots, Q_{n-1}$ are of the form $\pm(e_{i}\pm e_{j})$. Moreover, if $Q_{1}$ and $Q_{2}$ are of the form $\pm(e_{1}\pm e_{2})$, we have a contradiction similarly by dropping $e_{1}$ and $e_{2}$. Similarly, if $q(<k)$ vectors $Q_{j_{1}},\cdots, Q_{j_{q}}$ are spanned by $e_{1},\cdots,e_{q}$, we have a contradiction.

The above result is translated into that this graph with $k$ vertices and $k-1$ lines contains no loops described below;
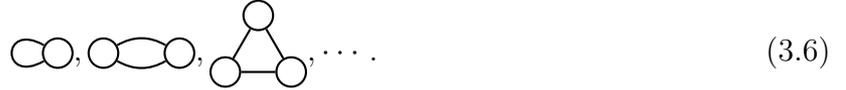
\begin{eqnarray}
\begin{tikzpicture}[auto,node distance=1cm,
  thick,main node/.style={circle,draw,minimum size=3mm},baseline=-.1cm]
  \node[main node] (1) {$$};
  \draw (1) to [in=150,out=210,loop] node {$$} (1);
\end{tikzpicture}
,
\begin{tikzpicture}[auto,node distance=1cm,
  thick,main node/.style={circle,draw,minimum size=3mm},baseline=-.1cm]
  \node[main node] (1) {$$};
  \node[main node] (2) [right of=1] {$$};
  \draw[bend right] (1) to node[very near start,below] {$$} node[very near end,below]{$$}(2);
  \draw[bend left] (1) to node[very near start,below] {$$} node[very near end,below]{$$}(2);
\end{tikzpicture}
,
\begin{tikzpicture}[auto,node distance=1cm,
  thick,main node/.style={circle,draw,minimum size=3mm},baseline=-.1cm]
  \node[main node] (1) at (90:.5cm) {$$};
  \node[main node] (2) at (210:.5cm) {$$};
  \node[main node] (3) at (330:.5cm) {$$};
  \draw (1) to node[very near start,below] {$$} node[very near end,below]{$$}(2);
  \draw (1) to node[very near start,below] {$$} node[very near end,below]{$$}(3);
  \draw (3) to node[very near start,below] {$$} node[very near end,below]{$$}(2);
\end{tikzpicture}
,\cdots.
\end{eqnarray}
In other words, the graph is made of trees\footnote{But we know for tree graphs
\begin{eqnarray}
\#(\text{connected component})=\#(\text{vertex})-\#(\text{line})=1
\end{eqnarray}
So we should have a connected tree diagram. 
}. 

The equation (\ref{eq1}) gives a system of $k$ linear equations. We start to solve this from the leaves of the graph to a fixed root. Finally, we have the following one relation
\begin{eqnarray}
\sum_{i=1}^{n} \pm \alpha^{p_{i}}=0
\end{eqnarray}
with appropriate signs. But $\alpha$ is large enough and then $\sum_{i=1}^{n} \pm \alpha^{p_{i}}\neq0$, which is a contradiction.

So the claim holds for the $n=l$ case, and then for any $n$.
\end{proof}
\subsection{Distinction rule}
To evaluate JK-residues, it is important to distinguish whether a basis $\sigma=(Q_{1},\cdots,Q_{n})$ of $\widetilde{{\bf Q}}^{BCD}$ is $\eta$-oriented. So here we give a distinction rule for a graph with the same number of lines and vertices and prove it. The resultant rule will be summarized as a theorem and a corollary in the last.

First we prove the following proposition that constraints the shape of the $\eta$-oriented graph.
\begin{proposition}
A graph with the same number of lines and vertices is $\eta$-oriented only if each of its connected components is expressed as a loop and trees associated to the vertices in the loop;
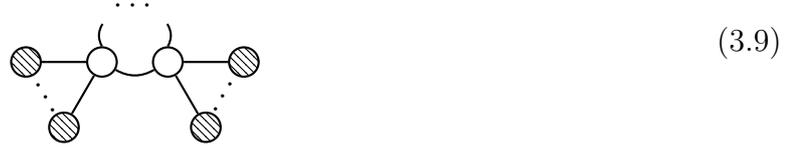
\begin{eqnarray}\label{loop}
\begin{tikzpicture}[auto,
  thick,main node/.style={circle,draw,minimum size=3mm},baseline=-.1cm]
  \node[] (1) at (90:.5cm) {$\cdots$};
  \node[main node] (2) at (210:.5cm) {$$};
  \node[main node] (3) at (330:.5cm) {$$};
  \node[main node,pattern=north west lines] (21) at ([shift=(180:1cm)]2){$$};
  \node[main node,pattern=north west lines] (22) at ([shift=(240:1cm)]2){$$};
  \node[main node,pattern=north west lines] (31) at ([shift=(0:1cm)]3){$$};
  \node[main node,pattern=north west lines] (32) at ([shift=(300:1cm)]3){$$};
  \node[] at ([shift=(120:.5cm)]22){\rotatebox{120}{$\cdots$}};
  \node[] at ([shift=(60:.5cm)]32){\rotatebox{60}{$\cdots$}};
  \draw[bend right] (1) to node[very near start,below] {$$} node[very near end,below]{$$}(2);
  \draw[bend right] (2) to node[very near start,below] {$$} node[very near end,below]{$$}(3);
  \draw[bend right] (3) to node[very near start,below] {$$} node[very near end,below]{$$}(1);
  \draw (2)--(21);
  \draw (2)--(22);
  \draw (3)--(31);
  \draw (3)--(32);
\end{tikzpicture}
\end{eqnarray}
where 
\begin{tikzpicture}[node distance=2cm,
  thick,main node/.style={circle,draw,minimum size=12pt},baseline=-3pt]
  \node[main node,pattern=north west lines] (1) {$$};
\end{tikzpicture}
means a tree graph. \hspace{\fill}$\blacksquare$
\end{proposition}
\begin{proof}
Take an $\eta$-oriented graph with $n$ lines $Q_{1},\cdots,Q_{n}\in\widetilde{{\bf Q}}^{BCD}$ and $n$ vertices. Then there is a solution of the system of $n$ linear equations
\begin{eqnarray}\label{eq2}
a_{1}Q_{1}+\cdots+a_{n}Q_{n}=\eta
\end{eqnarray}
with $a_{1},\cdots, a_{n}>0$.

If one connected component of the graph has more lines than vertices, we drop all vertices and lines in this component and get a subsystem of linear equations for the remaining graph  with more vertices than lines. But the previous discussion shows that this system and then (\ref{eq2}) have no solution. Therefore, each connected component of the $\eta$-oriented graph has the same number of lines as vertices, and then has just one loop. \end{proof}

\begin{remark}Dropping the lines which make the loop, we get tree subgraphs. For each tree subgraph, we assign the root of the tree to the unique vertex that is originally in the loop.
\end{remark}
What remains to get a distinction rule is to assign numbers and signs to the graph. To get intuition for the rule, we consider several examples below.
\begin{example}[{\bf 1. Straight Tree}]
Here we consider the $\eta$-oriented graph 
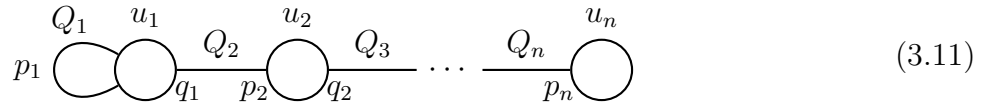
\begin{eqnarray}\label{treetypegraph}
\begin{tikzpicture}[auto,node distance=2cm,
  thick,main node/.style={circle,draw,minimum size=8mm},baseline=(current bounding box.center)]
  \node[main node] (1) {$$};
  \node[above] at (1.north){$u_{1}$};
  \node[main node] (2) [right of=1] {$$};
  \node[above] at (2.north){$u_{2}$};
  \node (c) [right of=2]{$\cdots$};
  \node[main node] (k) [right of=c]{$$};
  \node[above] at (k.north){$u_{n}$};
  \draw (1) to [in=150,out=210,loop] node {$p_{1}$}  node[near end, above]{$Q_{1}$}(1);
  \draw (1) to node[very near start,below] {$q_{1}$} node[above]{$Q_{2}$} node[very near end,below]{$p_{2}$}(2);
  \draw (2) to node[very near start,below] {$q_{2}$} node[above]{$Q_{3}$} node[very near end,below]{$$}(c);
  \draw (c) to node[very near start,below] {$$} node[above]{$Q_{n}$} node[very near end,below]{$p_{n}$}(k);
\end{tikzpicture}\ \ 
\end{eqnarray}
of $n$ vertices $u_{1},\cdots,u_{n}$ and $n$ lines $Q_{1},\cdots,Q_{n}$, where $p_{1},\cdots,p_{n}$ and $q_{1},\cdots,q_{n-1}$ are $\pm1$.

The distinction rule for this graph (\ref{treetypegraph}) is as follows.
\begin{lemma} The graph (\ref{treetypegraph}) is $\eta$-oriented if and only if $p_{n}=+1$ and
\begin{eqnarray}\label{distinction}
(p_{i},q_{i})=
\begin{cases}
(+,-),(-,+)&(M(u_{i})>m(u_{i}))\\
(+,\pm)&(M(u_{i})=m(u_{i})).
\end{cases}
\end{eqnarray}
for each $i<n$. \hspace{\fill}$\blacksquare$
\end{lemma}
\begin{proof}
Let us consider the $u_{n}$. We start to solve the system of linear equations (\ref{eq2}) from this vertex and we have
\begin{eqnarray}
p_{n}a_{n}=(\eta,e_{m(u_{n})})=\alpha^{m(u_{n})-1}.
\end{eqnarray}
Then we get $p_{n}=+1$ and $a_{n}=\alpha^{m(u_{n})-1}>0$.

Next we consider the vertex $u_{n-1}$ and we have
\begin{eqnarray}
p_{n-1}a_{n-1}=\alpha^{m(u_{n-1})-1}-q_{n}a_{n}=\alpha^{m(u_{n-1})-1}-q_{n}\alpha^{m(u_{n})-1}.
\end{eqnarray}
Since $\alpha$ is large enough, we get
\begin{align}
(p_{n-1},q_{n-1}; a_{n-1})&=
\begin{cases}
(+,-;\alpha^{m(u_{n-1})-1}+\alpha^{m(u_{n})-1}), (-,+;-\alpha^{m(u_{n-1})-1}\!\!\!\!\!\!&+\alpha^{m(u_{n})-1})\\&(m(u_{n})>m(u_{n-1}))\\
(+,\pm;\alpha^{m(u_{n-1})-1}\mp\alpha^{m(u_{n})-1})&(m(u_{n})<m(u_{n-1})).
\end{cases}
\end{align}
Especially, $a_{n-1}$ is of the order $\alpha^{M(u_{n-1})}$. 

Similarly, for a general vertex $u_{i}$, we have
\begin{eqnarray}
p_{i}a_{i}+q_{i}a_{i+1}=\alpha^{m(u_{i})-1}.
\end{eqnarray}
Note that $a_{i+1}$ is as large as $\alpha^{M(u_{i+1})}$. Then we get 
\begin{eqnarray}
(p_{i},q_{i})=
\begin{cases}
(+,-),(-,+)&(M(u_{i})>m(u_{i}))\\
(+,\pm)&(M(u_{i})=m(u_{i})).
\end{cases}
\end{eqnarray}
and $a_{i}$ is as large as $\alpha^{M(u_{i})}$.
\end{proof}

Note that one can see whether the signs are allowed for the $\eta$-oriented graph only by considering $M(u)$ and $m(u)$.
\end{example}
\begin{example}[{\bf 2. Loop}]
Next we consider a loop graph, 
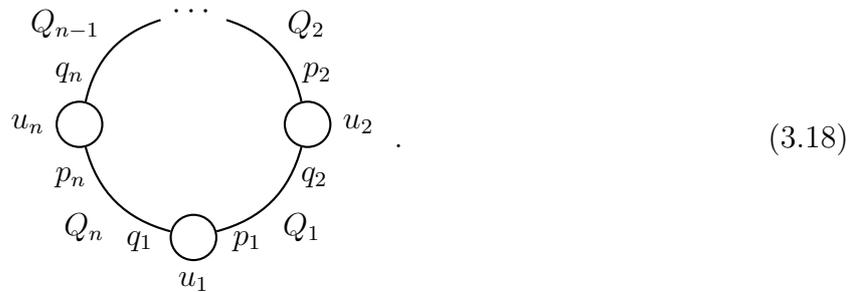
\begin{eqnarray}\label{desired}
\begin{tikzpicture}[auto,node distance=2cm,
  thick,main node/.style={circle,draw,minimum size=6mm},baseline=(current bounding box.center)]
  \node[] (1) at (90:1.5cm) {$\cdots$};
  \node[main node] (2) at (90+90:1.5cm) {$$};
  \node[main node] (3) at (90+90*2:1.5cm) {$$};
  \node[main node] (4) at (90+90*3:1.5cm) {$$};
  \node[left] at (2.west) {$u_{n}$};
  \node[below] at (3.south) {$u_{1}$};
  \node[right] at (4.east) {$u_{2}$};
  \draw[bend right] (1) to node[above left] {$Q_{n-1}$} node[very near start,below]{$$}node[near end,left]{$q_{n}$}(2);
  \draw[bend right] (2) to node[below left] {$Q_{n}$} node[near start,left]{$p_{n}$}node[near end,below]{$q_{1}$}(3);
  \draw[bend right] (3) to node[below right] {$Q_{1}$} node[near start,below]{$p_{1}$}node[near end,right]{$q_{2}$}(4);
  \draw[bend right] (4) to node[above right] {$Q_{2}$} node[near start,right]{$p_{2}$}node[very near end,below]{$$}(1);
\end{tikzpicture}\ .
\end{eqnarray}
with $n$ lines and $n$ vertices. We may assume $m(u_{1})=n>m(u_{i})$ for each $i\ge2$. 
\begin{lemma}
The graph (\ref{desired}) is $\eta$-oriented if and only if it is of the form
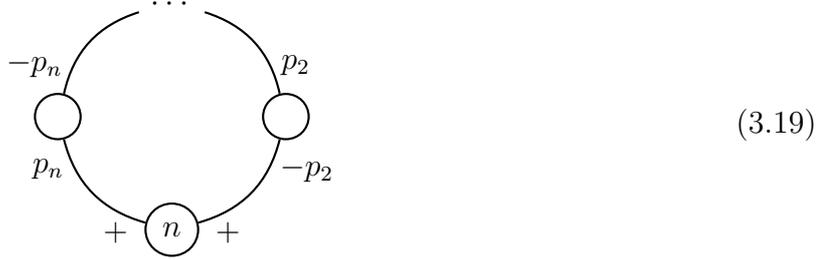
\begin{eqnarray}
\begin{tikzpicture}[auto,node distance=2cm,
  thick,main node/.style={circle,draw,minimum size=6mm},baseline=(current bounding box.center)]
  \node[] (1) at (90:1.5cm) {$\cdots$};
  \node[main node] (2) at (90+90:1.5cm) {$$};
  \node[main node] (3) at (90+90*2:1.5cm) {$n$};
  \node[main node] (4) at (90+90*3:1.5cm) {$$};
  \node[left] at (2.west) {$$};
  \node[below] at (3.south) {$$};
  \node[right] at (4.east) {$$};
  \draw[bend right] (1) to node[above left] {$$} node[very near start,below]{$$}node[near end,left]{$-p_{n}$}(2);
  \draw[bend right] (2) to node[below left] {$$} node[near start,left]{$p_{n}$}node[near end,below]{$+$}(3);
  \draw[bend right] (3) to node[below right] {$$} node[near start,below]{$+$}node[near end,right]{$-p_{2}$}(4);
  \draw[bend right] (4) to node[above right] {$$} node[near start,right]{$p_{2}$}node[very near end,below]{$$}(1);
\end{tikzpicture}
\end{eqnarray}
where $p_{i}=\pm1$ for each $2\le i\le n$.
\hspace{\fill}$\blacksquare$
\end{lemma}
\begin{proof}
Since $a_{1},a_{n}$ in (\ref{eq2}) are positive, either $p_{1}$ or $q_{1}$ is also positive. We may assume that $p_{1}=+1$ and $a_{1}$ is of the order $\alpha^{k-1}$ at least. If not, the system (\ref{eq2}) can not be satisfied for $a_{1},a_{n}>0$. 

Next we focus on the vertex $u_{2}$ and see that, in order to satisfy $a_{1}q_{2}+a_{2}p_{2}\le \alpha^{n-2}$, $p_{2}$ and $q_{2}$ have signs opposite to each other and $a_{2}$ is as large as $a_{1}$. 

Inductively, we can see that $p_{i}$ and $q_{i}$ have signs opposite to each other for each $1<i\le n$. Since the $n$ vectors $Q_{1},\cdots,Q_{n}$ becomes a basis of $\mathbb{R}^{n}$,
\begin{eqnarray}
0\neq|\mathrm{det}(Q_{1},\cdots,Q_{n})|=\left|\prod_{i=1}^{n}p_{i}+(-1)^{n-1}\prod_{i=1}^{n}q_{i}\right|=|1+p_{1}q_{1}|
\end{eqnarray}
must be satisfied. This determines $q_{1}=+1$. 

Conversely, if such assignments are given,the system (\ref{eq2}) gives
\begin{eqnarray}
a_{i}=p_{i}\alpha^{m(u_{i})-1}+a_{i-1} (i=2,\cdots,n)
\end{eqnarray}
and then
\begin{eqnarray}
a_{1}=\alpha^{n-1}-a_{n}=\alpha^{n-1}-a_{1}-\sum_{i=2}^{n}p_{i}\alpha^{m(u_{i})-1}.
\end{eqnarray}
So the system (\ref{eq2}) has a solution in the region where each $a_{i}\ (1\le i\le n)$ is as large as $\alpha^{n-1}$ and is necessarily positive. 
\end{proof}
\end{example}
\subsubsection*{Distinction rule}
Now we are ready to give a distinction rule for a general graph with the same number of lines and vertices. It is sufficient to give the rule for the graph (\ref{loop}). 

For the graph (\ref{loop}) with $n$ lines and $n$ vertices, the set $V$ of all the vertices are classified into the classes, $V=\sqcup_{i=1}^{5}V_{i}$, where
\begin{align}
V_{1}&=\{ u |\ u\ \text{is not in the loop and}\ m(v)=M(v).\}\\
V_{2}&=\{ u |\ u\ \text{is not in the loop and}\ m(v)<M(v).\}\\
V_{3}&=\{ u |\ u\ \text{is in the loop and}\ m(v)=M(v)=n.\}\\
V_{4}&=\{ u |\ u\ \text{is in the loop and}\ m(v)<M(v)=n.\}\\
V_{5}&=\{ u |\ u\ \text{is in the loop and}\ M(v)<n.\}.
\end{align}

For each vertex $u\in V_{1}\sqcup V_{2}$, we denote by $Q_{u}$ the line that connects $u$ and the parent of $u$, the unique vertex that links to $u$ by a line and is not descendant to $u$. Also for each vertex $u\in V_{2}\sqcup V_{4}$, we also denote by $T$ its descendant tree that has the vertex $v$ with $M(u)=m(v)$. 

Using the above notations, the distinction rule for the graph (\ref{loop}) is as follows.
\begin{theorem}[{\bf Distinction rule}]
The graph (\ref{loop}) is $\eta$-oriented if and only if all the following conditions are satisfied.

(i) For each $u\in V_{1}$, the assignment is of the form
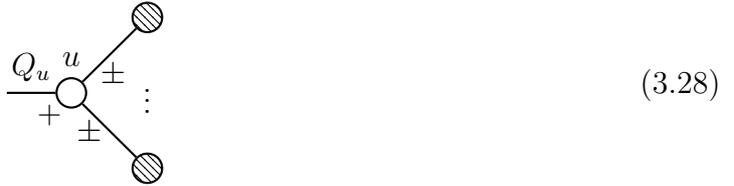
\begin{eqnarray}\label{ass1}
\begin{tikzpicture}[auto,node distance=1cm,
  thick,main node/.style={circle,draw},baseline=(current bounding box.center)]
  \node[main node] (1) {$$};
  \node[above] at (1.north){$u$};
  \node[] (0)[left of=1] {$$};
  \node[] (23) [right of=1] {$\vdots$};
  \node[main node,pattern=north west lines] (24) [below of=23] {$$};
  \node[main node,pattern=north west lines] (22) [above of=23] {$$};
  \draw (1) to node[very near start,below] {$+$} node[above]{$Q_{u}$} node[very near end,below]{$$}(0);
  \draw (1) to node[very near start,right] {$\pm$} node[above]{$$} node[very near end,below]{$$}(22);
   \draw (1) to node[very near start,below] {$\pm$} node[above]{$$} node[very near end,below]{$$}(24);
\end{tikzpicture} 
\end{eqnarray} with double signs for the trees in any order, where each \begin{tikzpicture}[auto,node distance=1cm,
  thick,main node/.style={circle,draw},baseline=-3pt]
  \node[main node,pattern=north west lines] (1) {$$};
\end{tikzpicture} means a tree whose vertices are descendant to $u$. 

(ii) For each $u\in V_{2}$, the assignment is of the form
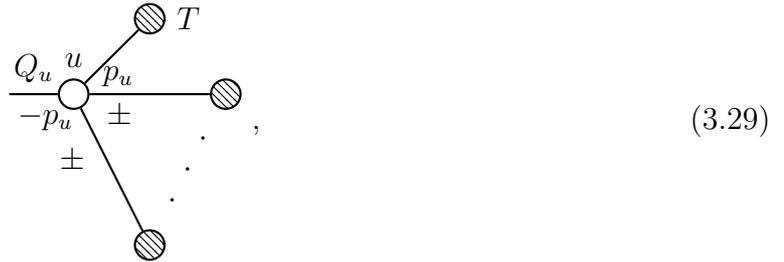
\begin{eqnarray} \label{ass2}
\begin{tikzpicture}[node distance=1cm,
  thick,main node/.style={circle,draw},baseline=(current bounding box.center)]
  \node[main node] (1) {$$};
  \node[] (0)[left of=1] {$$};  
  \node[] (22) [right of=1] {$$};   
  \node[main node,pattern=north west lines] (222) [right of=22] {$$};  
  \node[] (23) [below of=22] {$$};
  \node[main node,pattern=north west lines] (24) [below of=23] {$$};
  \node[main node,pattern=north west lines] (21) [above of=22] {$$};
  \node[right] at (21.east) {$T$};
  \node[above] at (1.north) {$u$};
  \draw (1) to node[near start,below] {$-p_{u}$} node[above]{$Q_{u}$} node[very near end,below]{$$}(0);
  \draw (1) to node[very near start,right] {$p_{u}$} node[above]{$$} node[very near end,below]{$$}(21); 
  \draw (1) to node[near start,below] {$\pm$} node[above]{$$} node[very near end,below]{$$}(222);
   \draw (1) to node[near start,below left] {$\pm$} node[above]{$$} node[very near end,below]{$$}(24);
   \draw[draw=none] (222) to node[pos=0.25]{$\cdot$}node[pos=0.5]{$\cdot$}node[pos=0.75]{$\cdot$} (24);
\end{tikzpicture}\ ,
\end{eqnarray}
with double signs in any order and $p_{u}=\pm1$.

(iii) For each vertex $u\in V_{3}$, the assignment is of the form
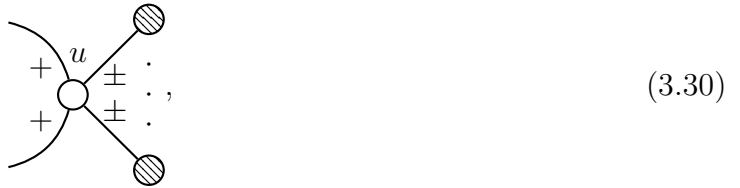
\begin{eqnarray}\label{general1}
\begin{tikzpicture}[node distance=1cm,
  thick,main node/.style={circle,draw},baseline=(current bounding box.center)]
  \node[] (2) at (180+90:1cm) {$$};
  \node[main node] (3) at (180+90*2:1cm) {$$};
  \node[above] at ([xshift=2pt,yshift=2pt]3.north) {$u$};
  \node[] (4) at (180+90*3:1cm) {$$};
  \draw[bend right] (2) to node[below left] {$$} node[near start,below]{$$}node[very near end,left]{$+$}(3);
  \draw[bend right] (3) to node[below right] {$$} node[very near start,left]{$+$}node[near end,above]{$$}(4);
  \node[] (33) [right of=3] {$$};
  \node[main node,pattern=north west lines] (34) [below of=33] {$$};
  \node[main node,pattern=north west lines] (32) [above of=33] {$$};
  \draw (3) to node[very near start,right] {$\pm$} node[above]{$$} node[very near end,below]{$$}(32);
   \draw (3) to node[very near start,right] {$\pm$} node[above]{$$} node[very near end,below]{$$}(34);
   \draw[draw=none] (34) to node[pos=0.25]{$\cdot$}node[pos=0.5]{$\cdot$}node[pos=0.75]{$\cdot$} (32);
\end{tikzpicture},
\end{eqnarray}
with double signs in any order, where the curves mean the lines in the loop.

(iv) For each vertex $u\in V_{4}$, the assignment is of the form
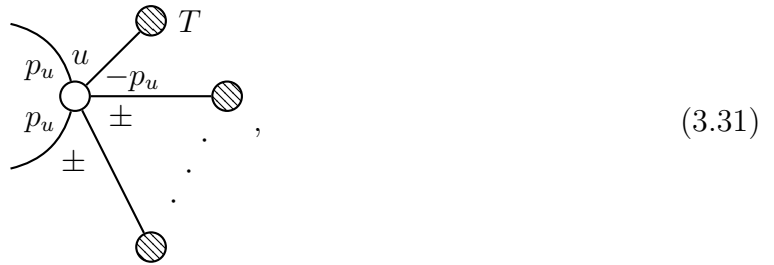
\begin{eqnarray}\label{general2}
\begin{tikzpicture}[node distance=1cm,
  thick,main node/.style={circle,draw},baseline=(current bounding box.center)]
  \node[] (2) at (180+90:1cm) {$$};
  \node[main node] (3) at (180+90*2:1cm) {$$};
  \node[] (4) at (180+90*3:1cm) {$$};
  \node[left] at (2.west) {$$};
  \node[above] at ([xshift=2pt,yshift=2pt]3.north) {$u$};
  \node[right] at (4.east) {$$};
  \draw[bend right] (2) to node[below left] {$$} node[near start,below]{$$}node[very near end,left]{$p_{u}$}(3);
  \draw[bend right] (3) to node[below right] {$$} node[very near start,left]{$p_{u}$}node[near end,above]{$$}(4);
   \node[] (32) [right of=3] {$$};   
  \node[main node,pattern=north west lines] (322) [right of=32] {$$};  
  \node[] (33) [below of=32] {$$};
  \node[] (34) [below of=32] {$$};
  \node[main node,pattern=north west lines] (35) [below of=34] {$$};
  \node[main node,pattern=north west lines] (31) [above of=32] {$$};
  \node[right] at (31.east) {$T$};
  \node[above] at (3.north) {$$};
  \draw (3) to node[very near start,right] {$-p_{u}$} node[above]{$$} node[very near end,below]{$$}(31); 
  \draw (3) to node[near start,below] {$\pm$} node[above]{$$} node[very near end,below]{$$}(322);
   \draw (3) to node[near start,below left] {$\pm$} node[above]{$$} node[very near end,below]{$$}(35);
    \draw[draw=none] (35) to node[pos=0.25]{$\cdot$}node[pos=0.5]{$\cdot$}node[pos=0.75]{$\cdot$} (322); 
\end{tikzpicture}\ ,
\end{eqnarray}
with double signs in any order and $p_{u}=\pm1$.

(v)  For each vertex $u\in V_{5}$, the assignment is of the form
\begin{eqnarray}\label{general3}
\begin{tikzpicture}[node distance=1cm,
  thick,main node/.style={circle,draw},baseline=(current bounding box.center)]
  \node[] (2) at (180+90:1cm) {$$};
  \node[main node] (3) at (180+90*2:1cm) {$$};
  \node[above] at ([xshift=2pt,yshift=2pt]3.north) {$u$};
  \node[] (4) at (180+90*3:1cm) {$$};
  \draw[bend right] (2) to node[below left] {$$} node[near start,below]{$$}node[very near end,left]{$p_{u}$}(3);
  \draw[bend right] (3) to node[below right] {$$} node[very near start,left]{$-p_{u}$}node[near end,above]{$$}(4);
  \node[] (33) [right of=3] {$$};
  \node[main node,pattern=north west lines] (34) [below of=33] {$$};
  \node[main node,pattern=north west lines] (32) [above of=33] {$$};
  \draw (3) to node[very near start,right] {$\pm$} node[above]{$$} node[very near end,below]{$$}(32);
   \draw (3) to node[very near start,right] {$\pm$} node[above]{$$} node[very near end,below]{$$}(34);
   \draw[draw=none] (34) to node[pos=0.25]{$\cdot$}node[pos=0.5]{$\cdot$}node[pos=0.75]{$\cdot$} (32);
\end{tikzpicture},
\end{eqnarray}
with double signs in any order and $p_{u}=\pm1$.\hspace{\fill}$\blacksquare$
\end{theorem}
\begin{proof}
(i)(ii) A similar discussion in Example 1 shows that the coefficient for the $Q_{u}$ in (\ref{eq2}) should be as large as $\alpha^{M(u)-1}$ and then gives the distinction rule for a vertex that is not in the loop. 

(iii)(iv)(v) A similar discussion in Example 2 gives the distinction rule for a vertex in the loop. 
\end{proof}
\begin{corollary}
A graph with the same number of lines and vertices is $\eta$-oriented if and only if each connected component of the graph is of the form (\ref{loop}) and satisfies the above distinction rule.\hspace{\fill}$\blacksquare$
\end{corollary}

\section{Jeffrey-Kirwan residue of $Sp(0)$ instantons}\label{box_arrangement}
\ytableausetup{centertableaux,boxsize=1.5em}
In the previous section, we have obtained the distinction rule for bases of $\widetilde{{\bf Q}}^{BCD}$ with respect to the special vector $\eta$. In this section, we will consider applications of the rule to the $BCD$ Nekrasov partition functions. Firstly, we introduce a tableau notation that represents Weyl orbits of the poles. This tableau specifies a class of graphs whose signs are relatively fixed and then we use the distinction rule to the graph to calculate the $A_{\sigma,\eta}$. This notation resembles the famous Young diagrammatic notation for the poles of the $U(N)$ Nekrasov partition function. For simplicity, we focus on the $G=Sp(N)$ case, but one can do parallel discussions for $G=SO(N)$.

Next, we will calculate JK-residues of ``$Sp(0)$" instantons and see how our results work in the calculation. In the appendix \ref{listcalc}, we present a list of all the JK-residue of $Sp(0)$ instanton up to $k=8$. As a result, we observe
\begin{eqnarray}
Z^{Sp(0)}_{k}=\frac{(-1)^{k}}{2^{k}k!\varepsilon^{k}_{1}\varepsilon^{k}_{2}}
\end{eqnarray}
for $k\le 8$.
\subsection{Box expressions for poles of $Sp(N)$ instantons}

For $G=Sp(N)$, there are $n$ variables to be integrated when we consider the Nekrasov partition function with instanton number $k=2n$ or $k=2n+1$. We want to know the JK-residue at $\phi_{*}=(\phi_{1*},\cdots,\phi_{n*})$.

Following (\ref{Qass}), one can obtain a graph of ${\bf Q}_{*}(\phi_{*})$. 
If the graph may not contain an $\eta$-oriented subgraph with $n$ lines\footnote{A trivial example is a graph with $|{\bf Q}_{*}|<n$.}, its JK-residue vanishes. A non-trivial pole should have an $\eta$-oriented retract. 

This condition significantly restricts the position of $\phi_{*}$ for the existence of the loop in (\ref{loop}). Each $\phi_{*j} (1\le j\le n)$ is confined on $2N+4$ lattices, 
\begin{eqnarray}
\begin{cases}
\pm a_{l}+\left(\mathbb{Z}+\frac{1}{2}\right)\varepsilon_{1}+\left(\mathbb{Z}+\frac{1}{2}\right)\varepsilon_{2}\quad (1\le l\le N)\\
\mathbb{Z}\varepsilon_{1}+\mathbb{Z}\varepsilon_{2},
(\mathbb{Z}+\frac{1}{2})\varepsilon_{1}+\mathbb{Z}\varepsilon_{2},
\mathbb{Z}\varepsilon_{1}+(\mathbb{Z}+\frac{1}{2})\varepsilon_{2},
(\mathbb{Z}+\frac{1}{2})\varepsilon_{1}+(\mathbb{Z}+\frac{1}{2})\varepsilon_{2}.
\end{cases}
\end{eqnarray}
Mirroring by the overall $\mathbb{Z}_{2}$ flip, we have the $N+4$ classes of the lattices. 

Now we rewrite such a pole $\phi_{*}=(\phi_{1*},\cdots,\phi_{n*})$ by an ($N+4$)-vector of tableaux, which has $k$ boxes in total. Firstly, put a box with number $i$ at $\phi_{i*}$ and a box with character $i'$ at $-\phi_{i*}$ for each $1\le i\le n$. Then, if $k$ is odd, put one additional box at 0 with number 0. 
\begin{example}
We write\footnote{We omit null shapes in the other classes of lattices.} the pole of 4-instanton at $\phi_{1*}=\frac{1}{2}\varepsilon_{1}+\varepsilon_{2}, \phi_{2*}=\frac{1}{2}\varepsilon_{1}$ by
\begin{eqnarray}
\begin{ytableau}
\none& 1'\\
2'&2\\
1
\end{ytableau}
\end{eqnarray}
and the one of 3-instanton at $\phi_{1*}=\varepsilon_{1}$ by
\begin{eqnarray}
\begin{ytableau}
1'&0&1
\end{ytableau}\ .
\end{eqnarray}
It is possible that multiple boxes are assigned at the same position. In this case, we should improve the expression in order to see the degeneracy. We express, for example, the pole of 4-instanton at $\phi_{1*}=0, \phi_{2*}=\varepsilon_{1}$ by
\begin{eqnarray}
\begin{ytableau}
2'&\scriptstyle 1,1'&2
\end{ytableau}\ ,
\end{eqnarray}
and one can see that there are two boxes at 0.
\hspace{\fill}$\blacksquare$
\end{example}

Forgetting numbers and characters but not its degeneracy, one can get an $(N+4)$-vector of diagrams for $\phi_{*}$. Clearly, this vector\footnote{This object first appeared in \cite{Marino:2004cn} and was called a generalized Young diagram in \cite{Hollands:2010xa}. } represents a Weyl orbit $\left[\phi_{*}\right]$ of poles. 

\begin{remark}[{\bf The connection between box expressions and graphs}]
The notation merely gives another expression of poles. One can construct the graph from its box adjacency. For example, 
\begin{eqnarray}
\begin{ytableau}
\none& 1'\\
2'&2\\
1
\end{ytableau}
\longrightarrow
\begin{tikzpicture}[auto,node distance=1cm,
  thick,main node/.style={circle,draw,minimum size=3mm},baseline=-.1cm]
  \node[main node] (1) {$$};
  \node[above] at (1.north){$2'$};
  \node[main node] (2) [right of=1] {$$};
  \node[below] at (2.south){$2$};
  \node[main node] (3) [below of=1] {$$};
  \node[below] at (3.south){$1$};  
  \node[main node] (4) [above of=2] {$$};
  \node[above] at (4.north){$1'$};
  \draw (1) to node[very near start,below] {$-$} node[above]{$$} node[very near end,below]{$+$}(2);
    \draw (1) to node[very near start,left] {$-$} node[above]{$$} node[very near end,left]{$+$}(3);
    \draw (2) to node[very near start,right] {$+$} node[above]{$$} node[very near end,right]{$-$}(4);
  \end{tikzpicture}
\xrightarrow{\text{Fold}}
\begin{tikzpicture}[auto,node distance=1cm,
  thick,main node/.style={circle,draw,minimum size=3mm},baseline=-.1cm]
  \node[main node] (1) {$$};
  \node[above] at (1.north){$2$};
  \node[main node] (2) [right of=1] {$$};
  \node[above] at (2.north){$1$};
  \draw (1) to [in=150,out=210,loop] node[] {$+$}  node[near end, above]{$$}(1);
  \draw (1) to node[very near start,below] {$+$} node[above]{$$} node[very near end,below]{$+$}(2);
  \end{tikzpicture}\ .
\end{eqnarray}
\end{remark}
\begin{remark}
We can see from this box expression that some poles have their vanishing JK-residues.  For example, we consider the diagram \ytableausetup{centertableaux,boxsize=0.6em}\ydiagram{2,1+2}\ \ytableausetup{centertableaux,boxsize=1.5em}for the case of 4-instanton. Thanks to the vanishing factor $\pm(\phi_{1}\pm\phi_{2})-\varepsilon$ in the numerator, this shape gives poles whose JK-residues are zero. 
For a general pole of $k(=2n,2n+1)$-instanton, from its diagrammatic expression, we can see the number $p$ of the zero factors in the numerator of the integrand and the number $q$ of those in the denominator. If $p-q>-n$, the JK-residue of the pole vanishes.
\end{remark}
\subsection{$Sp(0)$ instanton correction}
From now on, we concentrate on the pure SYM's instanton correction with gauge group ``$Sp(0)$", or the JK-residues of the rational function
\begin{eqnarray}
\mathcal{Z}_{k=2n+\chi}^{Sp(0)}\equiv\frac{1}{2}\frac{(-1)^{n}}{2^{n-1}n!}\left(\frac{\varepsilon}{\varepsilon_{1}\varepsilon_{2}}\right)^{n}\left(-\frac{1}{2\varepsilon_{1}\varepsilon_{2}}\right)^{\chi}\frac{1}{\prod_{j=1}^{n}(4\phi_{j}^{2}-\varepsilon_{1}^{2})(4\phi_{j}^{2}-\varepsilon_{2}^{2})}\frac{\Delta(0)\Delta(\varepsilon)}{\Delta(\varepsilon_{1})\Delta(\varepsilon_{2})}\ ,
\end{eqnarray}
where
\begin{eqnarray}
\Delta(x)=\left(\prod_{i<j}((\phi_{i}-\phi_{j})^{2}-x^{2})((\phi_{i}+\phi_{j})^{2}-x^{2})\right)\left(\prod_{j=1}^{n}(\phi_{j}^{2}-x^{2})\right)^{\chi}\ .
\end{eqnarray}
Note that the integrand $\mathcal{Z}_{k}^{Sp(0)}$ is invariant under the Weyl group $\mathfrak{S}_{n}\ltimes \mathbb{Z}^{n}_{2}$ action on $(\phi_{1},\cdots,\phi_{n})$. Then we have
\begin{eqnarray}
Z^{Sp(0)}_{k}\equiv\sum_{\phi_{*}}\mathrm{JK\text{-}Res}_{\eta,\phi_{*}}\mathcal{Z}^{Sp(0)}_{k}=\sum_{[\phi_{*}]}\frac{1}{C_{[\phi_{*}]}}\sum_{\sigma\in\mathfrak{B}({\bf Q}_{*})}\frac{A_{\sigma,\eta}c_{\sigma}(\phi_{*})}{|\mathrm{det}(Q^{\sigma}_{i}\cdot e_{j})|},
\end{eqnarray}
where $C_{[\phi_{*}]}$ is the order of the stabilizer group at $\phi_{*}$ and 
\begin{align}
A_{\sigma,\eta}&\equiv\sum_{g\in S_{n}\ltimes\mathbb{Z}^{n}_{2}}{\bf 1}_{\mathrm{Cone}(g\cdot\sigma)}(\eta)\in\mathbb{N},\\
\mathcal{Z}^{Sp(0)}_{k}\in\hat{R}_{{\bf Q}_{*},\phi_{*}}&\mapsto\sum_{\sigma\in\mathfrak{B}({\bf Q}_{*})}c_{\sigma}(\phi_{*})f_{\sigma}\in S_{{\bf Q}_{*},\phi_{*}}.\label{extraction}
\end{align}
Here the map in the last is the part of the definition of the JK-residue (\ref{JKJK}).

A Weyl orbit $[\phi_{*}]$ that has non-vanishing JK-residue should be expressed as a 4-vector of diagrams with $k$ boxes in total. Each component belongs to each of the four lattices 
\begin{eqnarray}
\mathbb{Z}\varepsilon_{1}+\mathbb{Z}\varepsilon_{2}, \left(\mathbb{Z}+\frac{1}{2}\right)\varepsilon_{1}+\mathbb{Z}\varepsilon_{2}, \mathbb{Z}\varepsilon_{1}+\left(\mathbb{Z}+\frac{1}{2}\right)\varepsilon_{2}, \left(\mathbb{Z}+\frac{1}{2}\right)\varepsilon_{1}+\left(\mathbb{Z}+\frac{1}{2}\right)\varepsilon_{2}.
\end{eqnarray}
We do the extraction (\ref{extraction}) for each non-trivial orbit and calculate the $A_{\sigma,\eta}$ for each basis $\sigma$ with $c_{\sigma}(\phi_{*})\neq0$ by applying the distinction rule to the set. Here we set $\eta=(1,\alpha,\cdots,\alpha^{n-1})$ for a fixed sufficient large $\alpha$ and then the distinction rule determines the $A_{\sigma,\eta}$.
 
Also, we can determine the factor $|\mathrm{det}( Q^{\sigma}_{i}\cdot e_{j})|$. For a given class of subgraphs, each connected component contributes to this factor by $1$ if its loop is of the form \!\!\!\!\!
\begin{tikzpicture}[node distance=.5cm,
  thick,main node/.style={circle,draw,minimum size=1mm},baseline=-.1cm]
  \node[main node] (1) {$$};
  \draw (1) to [in=150,out=210,loop] node {$$}  node[near end, above]{$$}(1);
  \end{tikzpicture},
or by $2$ otherwise.

What remains to determine $Z^{Sp(0)}_{k,\eta}$ is $c_{\sigma}(\phi_{*})$, the decomposition coefficients of $\mathcal{Z}^{Sp(0)}_{k}$ into the basic fractions. {\it At the present stage, however, we have to calculate them explicitly.} 
\begin{example}[{\bf An explicit JK-residue calculation}]
Here we give an example of the JK-residue calculation in order to see the algorithm more clearly.

Let us consider the pole of $\mathcal{Z}^{Sp(0)}_{6}$ at $\phi_{*1}=\phi_{*2}=\frac{1}{2}\varepsilon_{1}, \phi_{*3}=\frac{3}{2}\varepsilon_{1}$. Then its orbit is expressed by \ytableausetup{centertableaux,boxsize=.6em}\ydiagram[*(gray)]{1+2}*[*(white)]{4}\ytableausetup{centertableaux,boxsize=1.5em}\ , where each gray box represents two boxes. We see $C_{[\phi_{*}]}=2$ from its overlap. 

From the numerator, we count $p=2$ for $\phi_{*1}-\phi_{*2}=0$. On the other hand, from the denominator, we count $q=5$ from $\phi_{*1}-\frac{1}{2}\varepsilon_{1}=\phi_{*2}-\frac{1}{2}\varepsilon_{1}=\phi_{*1}+\phi_{*2}-\varepsilon_{1}=0$ and $\phi_{*3}-\phi_{*1}-\varepsilon_{1}=\phi_{*3}-\phi_{*2}-\varepsilon_{1}=0$. So this orbit may have non-zero JK-residue for $2-5=-3=-n$. Especially, the total degree of zero factors at $\phi_{*}$ is $-3$, and then we expand the remain of the integrand at the order 0, which gives
\begin{eqnarray}
c=-\frac{(2\varepsilon_{1}+\varepsilon_{2})(3\varepsilon_{1}+\varepsilon_{2})}{27648\varepsilon_{1}^{6}(\varepsilon_{1}-\varepsilon_{2})^{3}(2\varepsilon_{1}-\varepsilon_{2})^{2}(3\varepsilon_{1}-\varepsilon_{2})\varepsilon_{2}^{2}}.
\end{eqnarray}

The corresponding class of graphs for the orbit $[\phi_{*}]$ is
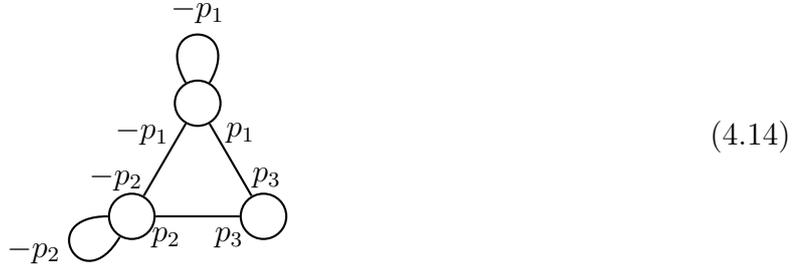
\begin{eqnarray}
\begin{tikzpicture}[node distance=2cm,
  thick,main node/.style={circle,draw,minimum size=6mm},baseline=(current bounding box.center)]
  \node[main node] (1) at (90:1cm) {$$};
  \node[main node] (2) at (210:1cm) {$$};
  \node[main node] (3) at (330:1cm) {$$};
  \draw (1) to node[very near start,left] {$-p_{1}$} node[near end,left]{$-p_{2}$}(2);
  \draw (1) to node[very near start,right] {$p_{1}$} node[near end,right]{$p_{3}$}(3);
  \draw (3) to node[very near start,below] {$p_{3}$} node[very near end,below]{$p_{2}$}(2);
  \draw (1) to [in=60,out=120,loop] node[above] {$-p_{1}$}  node[near end, above]{$$}(1);
  \draw (2) to [in=180,out=240,loop] node[left] {$-p_{2}$}  node[near end, above]{$$}(2);
\end{tikzpicture}
\end{eqnarray} 
where $p_{1},p_{2}$ and $p_{3}$ are $\pm1$. The Weyl group action ends up with flipping these signs and permutating the assigned numbers. 

Note that the zero factors in the numerator become
\begin{eqnarray}
(\phi_{1}-\phi_{2})^{2}=-\left((\phi_{1}-\frac{\varepsilon_{1}}{2})-(\phi_{2}-\frac{\varepsilon_{1}}{2})\right)\left((\phi_{3}-\phi_{1}-\varepsilon_{1})-(\phi_{3}-\phi_{2}-\varepsilon_{1})\right)\ ,
\end{eqnarray}
then we have a decomposition of the zero factors into some basic fractions
\begin{eqnarray}
-f_{\sigma_{1}}+f_{\sigma_{2}}+f_{\sigma_{3}}-f_{\sigma_{4}} 
\end{eqnarray}
where 
\begin{eqnarray}
\sigma_{1}=\!\!\!\!\!\!\!
\begin{tikzpicture}[node distance=2cm,
  thick,main node/.style={circle,draw,minimum size=6mm},baseline=(current bounding box.center)]
  \node[main node] (1) at (90:1cm) {$$};
  \node[main node] (2) at (210:1cm) {$$};
  \node[main node] (3) at (330:1cm) {$$};
  \draw (1) to node[very near start,left] {$-p_{1}$} node[near end,left]{$-p_{2}$}(2);
  \draw (3) to node[very near start,below] {$p_{3}$} node[very near end,below]{$p_{2}$}(2);
  \draw (2) to [in=180,out=240,loop] node[left] {$-p_{2}$}  node[near end, above]{$$}(2);
\end{tikzpicture}\ ,
\sigma_{2}=\!\!\!\!\!\!\!
\begin{tikzpicture}[node distance=2cm,
  thick,main node/.style={circle,draw,minimum size=6mm},baseline=(current bounding box.center)]
  \node[main node] (1) at (90:1cm) {$$};
  \node[main node] (2) at (210:1cm) {$$};
  \node[main node] (3) at (330:1cm) {$$};
  \draw (1) to node[very near start,left] {$-p_{1}$} node[near end,left]{$-p_{2}$}(2);
  \draw (1) to node[very near start,right] {$p_{1}$} node[near end,right]{$p_{3}$}(3);
  \draw (2) to [in=180,out=240,loop] node[left] {$-p_{2}$}  node[near end, above]{$$}(2);
\end{tikzpicture}\ ,
\sigma_{3}=\!\!\!\!\!\!\!
\begin{tikzpicture}[node distance=2cm,
  thick,main node/.style={circle,draw,minimum size=6mm},baseline=(current bounding box.center)]
  \node[main node] (1) at (90:1cm) {$$};
  \node[main node] (2) at (210:1cm) {$$};
  \node[main node] (3) at (330:1cm) {$$};
  \draw (1) to node[very near start,left] {$-p_{1}$} node[near end,left]{$-p_{2}$}(2);
  \draw (3) to node[very near start,below] {$p_{3}$} node[very near end,below]{$p_{2}$}(2);
  \draw (1) to [in=60,out=120,loop] node[above] {$-p_{1}$}  node[near end, above]{$$}(1);
\end{tikzpicture}\ ,
\sigma_{4}=\!\!\!\!\!\!\!
\begin{tikzpicture}[node distance=2cm,
  thick,main node/.style={circle,draw,minimum size=6mm},baseline=(current bounding box.center)]
  \node[main node] (1) at (90:1cm) {$$};
  \node[main node] (2) at (210:1cm) {$$};
  \node[main node] (3) at (330:1cm) {$$};
  \draw (1) to node[very near start,left] {$-p_{1}$} node[near end,left]{$-p_{2}$}(2);
  \draw (1) to node[very near start,right] {$p_{1}$} node[near end,right]{$p_{3}$}(3);
  \draw (1) to [in=60,out=120,loop] node[above] {$-p_{1}$}  node[near end, above]{$$}(1);
\end{tikzpicture}\ .
\end{eqnarray}
Using the distinction rule, we have 
\begin{eqnarray}
A_{\sigma_{1},\eta}=A_{\sigma_{4},\eta}=6,\ A_{\sigma_{2},\eta}=A_{\sigma_{3},\eta}=3.
\end{eqnarray}
Thus we get the sum of the JK-residues over the orbit $[\phi_{*}]$
\begin{align}
\sum_{\psi_{*}\in[\phi_{*}]}\mathrm{JK\text{-}Res}_{\eta,\psi_{*}}\mathcal{Z}^{Sp(0)}_{6}&=\frac{c}{2}\left(-6+3+3-6\right)\nonumber\\&=\frac{(2\varepsilon_{1}+\varepsilon_{2})(3\varepsilon_{1}+\varepsilon_{2})}{9216\varepsilon_{1}^{6}(\varepsilon_{1}-\varepsilon_{2})^{3}(2\varepsilon_{1}-\varepsilon_{2})^{2}(3\varepsilon_{1}-\varepsilon_{2})\varepsilon_{2}^{2}}\ .
\end{align}
\hspace{\fill}$\blacksquare$
\end{example}
\begin{remark}
In general, the calculation becomes easier if one considers a pole $\phi_{*}\in\mathbb{C}^{n}$ where the integrand has no terms below the degree $-n$ at $\phi_{*}$, because one needs only the zeroth term of the Taylor series of the nonzero factors.
\end{remark}
We can do similar calculations for other Weyl orbits. We give a list of the explicit residues for $k\le 8$ in the appendix \ref{listcalc}. Finally, we observed the following property:
\begin{observation}
\begin{eqnarray}\label{exponential}
Z^{Sp(0)}_{k,\eta}=\frac{(-1)^{k}}{2^{k}k!\varepsilon^{k}_{1}\varepsilon^{k}_{2}}
\end{eqnarray}
at least for $k\le 8$.
\hspace{\fill}$\blacksquare$
\end{observation}
This observation seems to be generalized for any $k$ and looks to be an analogue of the combinational formula $Z^{U(1)}_{k}=(k!\varepsilon^{k}_{1}\varepsilon^{k}_{2})^{-1}$. 

\section{Discussions}\label{discuss}
In this paper, we have made an algorithm to calculate the $BCD$-type Nekrasov partition functions more easily and then observed (\ref{exponential}) for the $Sp(0)$ case for small instanton numbers. One can determine the residue of a Weyl orbit of poles, or a diagram without considering awkward cancellations of iterated residues, but with the distinction rule we have proposed. We hope that this algorithm will help readers to calculate some multi-contour integrals. 

However, we have calculated residues mainly by substituting positions of poles for the nonzero factors in the integrand. We have calculated the JK-residues of 106 Weyl orbits for the $Sp(0)$ 8-instanton correction and found that the 104 orbits contributes to the correction. For larger instanton numbers, one should encounter more orbits. Then it is desirable to express them only by some manipulations of box diagrams, like the formula with Young diagrams for the $U(N)$ case. Also, we have observed the two nontrivial orbits (\ref{nontrivpole1}) and (\ref{nontrivpole2}) whose joint residues vanish. We do not confirm the set of all the orbits with nonzero joint residues. One way to tackle with this problem is to construct a recursion relation among the fixed points, or to understand how the residues changes when one adds boxes to the poles. In this regard, our distinction rule may play a role to calculate the ratio of two $A_{\sigma,\eta}$ through a graphical consideration. However, the relation is not clear because one sometimes needs to differentiate the nonzero part of the integrand at a pole to evaluate its JK-residue. For example, we did such a differentiation when we evaluated (\ref{doublepole}). This differentiation makes it difficult to rephrase the recursion relation in terms of box manipulations.

This simple form of the sum for the $U(1)$ Nekrasov partition function can be understood from the perspective of Hilbert schemes of points on $\mathbb{C}^{2}$\cite{nakajimalectures,Nakajima:2003pg}. Roughly, this property is described as
\begin{eqnarray}
\int_{\mathbb{C}^{2k}/\mathfrak{S}_{k}} 1=\frac{1}{k!\varepsilon^{k}_{1}\varepsilon^{k}_{2}}.
\end{eqnarray}
Moreover, the residue formula for the $U(1)$ case can be expressed as the combinational formula of Jack polynomials. Maybe there are a similar background behinds the $Sp(0)$ case. Are there an analog $X_{k}$ which gives
\begin{eqnarray}
\int_{X_{k}} 1=\frac{(-1)^{k}}{2^{k}k!\varepsilon^{k}_{1}\varepsilon^{k}_{2}}
\end{eqnarray} 
and corresponding orthogonal polynomials? If one can show it explicitly, then it will enable ones to explore how outer automorphisms of Lie algebras act on instantons with different gauge groups, which may lead to deeper understandings of gauge theories with exceptional gauge groups because any exceptional Lie algebras relate to some classical Lie algebras through such automorphisms. 

Physically, the $U(N)$ Nekrasov partition function can be realized from an D0-D4 system. Similarly, an D0-O4 system shall realize the $Sp(0)$ situation. It may be interesting to understand the $Z^{Sp(0)}_{k}$ from the string theory. Moreover, the AGT correspondence suggests that some symmetries arising from 2d CFT act on poles of the Nekrasov partition functions. One may construct the AGT proof for $BCD$ cases through understanding their pole distributions. 

\section*{Acknowledgements}
The author thanks Yutaka Matsuo, Kantaro Ohmori and Futoshi Okazawa for many helpful discussions. SN is also grateful to Yuji Tachikawa for guiding SN to the notion of the Jeffrey-Kirwan residue as well as for helpful discussions.

\appendix

\section{A list of the Jeffrey-Kirwan residues of the $Sp(0)$ Nekrasov partition function}\label{listcalc}
\ytableausetup{centertableaux,boxsize=.6em}
Here we give a list of all the JK-residues of the $Sp(0)$ Nekrasov partition functions $Z^{Sp(0)}_{k}$ with instanton number $k\le 8$. We denote the sum of JK-residues over a Weyl orbit $[\phi_{*}]$ by $Z^{Sp(0)}_{k}([\phi_{*}])$. We denote two boxes by \ydiagram[*(gray)]{1}, and there boxes by \begin{ytableau}\times\end{ytableau}. We omit the poles whose residues vanishes trivially. From the list, one can get $Z^{Sp(0)}_{k}=\frac{(-1)^{k}}{2^{k}k!\varepsilon^{k}_{1}\varepsilon^{k}_{2}}$ for $k\le 8$.
\ytableausetup{centertableaux,boxsize=.6em}
\subsection{$k=0,1,2,3$}
\begin{align}
Z^{Sp(0)}_{0}\left(\cdot\right)&=1\\
Z^{Sp(0)}_{1}\left(\ydiagram{1}\right)&=-\frac{1}{2\varepsilon_{1}\varepsilon_{2}}\\
Z^{Sp(0)}_{2}\left(\ydiagram{2}\right)&=-\frac{1}{8\varepsilon^{2}_{1}\varepsilon_{2}(\varepsilon_{1}-\varepsilon_{2})}\\
Z^{Sp(0)}_{2}\left(\ \ydiagram{1,1}\ \right)&=-\frac{1}{8\varepsilon^{}_{1}\varepsilon^{2}_{2}(\varepsilon_{1}-\varepsilon_{2})}\\
Z^{Sp(0)}_{3}\left(\ydiagram{3}\right)&=-\frac{1}{24\varepsilon^{3}_{1}\varepsilon^{}_{2}(\varepsilon_{1}-\varepsilon_{2})(2\varepsilon_{1}-\varepsilon_{2})}\\
Z^{Sp(0)}_{3}\left(\ \ydiagram{1,1,1}\ \right)&=-\frac{1}{24\varepsilon^{}_{1}\varepsilon^{3}_{2}(\varepsilon_{1}-\varepsilon_{2})(\varepsilon_{1}-2\varepsilon_{2})}\\
Z^{Sp(0)}_{3}\left(\ \ydiagram{2}\ ,\ydiagram{1}\ \right)&=\frac{3\varepsilon_{1}+2\varepsilon_{2}}{48\varepsilon^{3}_{1}\varepsilon^{2}_{2}(\varepsilon_{1}-\varepsilon_{2})(\varepsilon_{1}-2\varepsilon_{2})}\\
Z^{Sp(0)}_{3}\left(\ \ydiagram{1,1}\ ,\ydiagram{1}\ \right)&=\frac{2\varepsilon_{1}+3\varepsilon_{2}}{48\varepsilon^{2}_{1}\varepsilon^{3}_{2}(\varepsilon_{1}-\varepsilon_{2})(2\varepsilon_{1}-\varepsilon_{2})}
\end{align}
\subsection{$k=4$}
\begin{align}
Z^{Sp(0)}_{4}\left(\ \ydiagram{1,1,1,1}\ \right)&=\frac{1}{128 \varepsilon_{1} (\varepsilon_{1}-3 \varepsilon_{2}) (\varepsilon_{1}-2 \varepsilon_{2})
   (\varepsilon_{1}-\varepsilon_{2}) \varepsilon_{2}^4}\\
Z^{Sp(0)}_{4}\left(\ \ydiagram{1+2,2}\ \right)&=\frac{2 \varepsilon_{1}+\varepsilon_{2}}{64 \varepsilon_{1}^2 (\varepsilon_{1}-2 \varepsilon_{2})
   (\varepsilon_{1}-\varepsilon_{2})^2 (2 \varepsilon_{1}-\varepsilon_{2}) (3 \varepsilon_{1}-\varepsilon_{2})
   \varepsilon_{2}^2}\\
Z^{Sp(0)}_{4}\left(\ \ydiagram{2}\ ,\ydiagram{1,1}\ \right)&=-\frac{3 (\varepsilon_{1}+\varepsilon_{2})^2}{64 \varepsilon_{1}^3 (\varepsilon_{1}-3 \varepsilon_{2})
   (\varepsilon_{1}-\varepsilon_{2})^2 (3 \varepsilon_{1}-\varepsilon_{2}) \varepsilon_{2}^3}\\
Z^{Sp(0)}_{4}\left(\ \ydiagram{1+1,2,1}\ \right)&=\frac{-\varepsilon_{1}-2 \varepsilon_{2}}{64 \varepsilon_{1}^2 (\varepsilon_{1}-3 \varepsilon_{2})
   (\varepsilon_{1}-2 \varepsilon_{2}) (\varepsilon_{1}-\varepsilon_{2})^2 (2 \varepsilon_{1}-\varepsilon_{2})
   \varepsilon_{2}^2}\\
Z^{Sp(0)}_{4}\left(\ \ydiagram{2,2}\ \right)&=\frac{1}{32 \varepsilon_{1}^2 (\varepsilon_{1}-2 \varepsilon_{2}) (\varepsilon_{1}-\varepsilon_{2})^2 (2
   \varepsilon_{1}-\varepsilon_{2}) \varepsilon_{2}^2}\\
Z^{Sp(0)}_{4}\left(\ \ydiagram[*(gray)]{0,1}*[*(white)]{1,1,1}\ \right)&=\frac{-\varepsilon_{1}-2 \varepsilon_{2}}{192 \varepsilon_{1}^2 (\varepsilon_{1}-2 \varepsilon_{2})
   (\varepsilon_{1}-\varepsilon_{2})^2 \varepsilon_{2}^4}\\
Z^{Sp(0)}_{4}\left(\ \ydiagram{4}\ \right)&=-\frac{1}{128 \varepsilon_{1}^4 (\varepsilon_{1}-\varepsilon_{2}) (2 \varepsilon_{1}-\varepsilon_{2}) (3
   \varepsilon_{1}-\varepsilon_{2}) \varepsilon_{2}}\\
Z^{Sp(0)}_{4}\left(\ \ydiagram[*(gray)]{1+1}*[*(white)]{3}\ \right)&=\frac{2 \varepsilon_{1}+\varepsilon_{2}}{192 \varepsilon_{1}^4 (\varepsilon_{1}-\varepsilon_{2})^2 (2
   \varepsilon_{1}-\varepsilon_{2}) \varepsilon_{2}^2}
\end{align}
\subsection{$k=5$}
\begin{align}
 Z^{Sp(0)}_{5}\left(\ \ydiagram{5}\ \right)&=-\frac{1}{640 \varepsilon_{1}^5 (\varepsilon_{1}-\varepsilon_{2}) (2 \varepsilon_{1}-\varepsilon_{2}) (3
   \varepsilon_{1}-\varepsilon_{2}) (4 \varepsilon_{1}-\varepsilon_{2}) \varepsilon_{2}} \\
 Z^{Sp(0)}_{5}\left(\ \ydiagram{4}\ ,\ydiagram{1}\ \right)&=\frac{3 (5 \varepsilon_{1}+2 \varepsilon_{2})}{1280 \varepsilon_{1}^5 (3 \varepsilon_{1}-2 \varepsilon_{2})
   (\varepsilon_{1}-\varepsilon_{2}) (2 \varepsilon_{1}-\varepsilon_{2}) (3 \varepsilon_{1}-\varepsilon_{2})
   \varepsilon_{2}^2} \\
 Z^{Sp(0)}_{5}\left(\ \ydiagram{3}\ ,\ydiagram{2}\ \right)&=\frac{(3 \varepsilon_{1}+2 \varepsilon_{2}) (5 \varepsilon_{1}+2 \varepsilon_{2})}{960 \varepsilon_{1}^5
   (\varepsilon_{1}-2 \varepsilon_{2}) (3 \varepsilon_{1}-2 \varepsilon_{2}) (\varepsilon_{1}-\varepsilon_{2})^2 (2
   \varepsilon_{1}-\varepsilon_{2}) \varepsilon_{2}^2} \\
 Z^{Sp(0)}_{5}\left(\ \ydiagram{2+1,3,1}\ \right)&=\frac{\varepsilon_{1}+2 \varepsilon_{2}}{64 \varepsilon_{1}^3 (2 \varepsilon_{1}-3 \varepsilon_{2})
   (\varepsilon_{1}-2 \varepsilon_{2}) (3 \varepsilon_{1}-2 \varepsilon_{2}) (\varepsilon_{1}-\varepsilon_{2})^2 (2
   \varepsilon_{1}-\varepsilon_{2}) (3 \varepsilon_{1}-\varepsilon_{2}) \varepsilon_{2}} \\
 Z^{Sp(0)}_{5}\left(\ \ydiagram{3}\ ,\ydiagram{1,1}\ \right)&=-\frac{(2 \varepsilon_{1}+\varepsilon_{2}) (4 \varepsilon_{1}+3 \varepsilon_{2})}{192 \varepsilon_{1}^4 (2
   \varepsilon_{1}-3 \varepsilon_{2}) (\varepsilon_{1}-\varepsilon_{2})^2 (2 \varepsilon_{1}-\varepsilon_{2}) (4
   \varepsilon_{1}-\varepsilon_{2}) \varepsilon_{2}^3} \\
 Z^{Sp(0)}_{5}\left(\ \ydiagram{1+1,2,1}\ ,\ydiagram{1}\ \right)&=\frac{(3 \varepsilon_{1}+2 \varepsilon_{2}) (\varepsilon_{1}+4 \varepsilon_{2})}{128 \varepsilon_{1}^3
   (\varepsilon_{1}-4 \varepsilon_{2}) (\varepsilon_{1}-3 \varepsilon_{2}) (3 \varepsilon_{1}-2 \varepsilon_{2})
   (\varepsilon_{1}-\varepsilon_{2})^2 (2 \varepsilon_{1}-\varepsilon_{2}) \varepsilon_{2}^3} \\
 Z^{Sp(0)}_{5}\left(\ \ydiagram{2}\ ,\ydiagram{1,1,1}\ \right)&=\frac{(\varepsilon_{1}+2 \varepsilon_{2}) (3 \varepsilon_{1}+4 \varepsilon_{2})}{192 \varepsilon_{1}^3
   (\varepsilon_{1}-4 \varepsilon_{2}) (\varepsilon_{1}-2 \varepsilon_{2}) (3 \varepsilon_{1}-2 \varepsilon_{2})
   (\varepsilon_{1}-\varepsilon_{2})^2 \varepsilon_{2}^4} \\
 Z^{Sp(0)}_{5}\left(\ \ydiagram{2}\ ,\ydiagram{1,1}\ ,\ydiagram{1}\ \right)&=-\frac{(\varepsilon_{1}+\varepsilon_{2})^2 (3 \varepsilon_{1}+2 \varepsilon_{2}) (2 \varepsilon_{1}+3
   \varepsilon_{2})}{384 \varepsilon_{1}^4 (\varepsilon_{1}-3 \varepsilon_{2}) (\varepsilon_{1}-2 \varepsilon_{2})
   (\varepsilon_{1}-\varepsilon_{2})^2 (2 \varepsilon_{1}-\varepsilon_{2}) (3 \varepsilon_{1}-\varepsilon_{2})
   \varepsilon_{2}^4} \\
Z^{Sp(0)}_{5}\left(\ \ydiagram{1+2,1+1,2}\ \right)&=-\frac{2 \varepsilon_{1}+\varepsilon_{2}}{64 \varepsilon_{1} (\varepsilon_{1}-3 \varepsilon_{2}) (2
   \varepsilon_{1}-3 \varepsilon_{2}) (\varepsilon_{1}-2 \varepsilon_{2}) (3 \varepsilon_{1}-2 \varepsilon_{2})
   (\varepsilon_{1}-\varepsilon_{2})^2 (2 \varepsilon_{1}-\varepsilon_{2}) \varepsilon_{2}^3} \\
 Z^{Sp(0)}_{5}\left(\ \ydiagram{2,2}\ ,\ydiagram{1}\ \right)&=-\frac{3 (\varepsilon_{1}+\varepsilon_{2})^2}{64 \varepsilon_{1}^3 (\varepsilon_{1}-3 \varepsilon_{2})
   (\varepsilon_{1}-2 \varepsilon_{2}) (\varepsilon_{1}-\varepsilon_{2})^2 (2 \varepsilon_{1}-\varepsilon_{2}) (3
   \varepsilon_{1}-\varepsilon_{2}) \varepsilon_{2}^3} \\
 Z^{Sp(0)}_{5}\left(\ \ydiagram{1+2,2}\ ,\ydiagram{1}\ \right)&=\frac{(4 \varepsilon_{1}+\varepsilon_{2}) (2 \varepsilon_{1}+3 \varepsilon_{2})}{128 \varepsilon_{1}^3 (2
   \varepsilon_{1}-3 \varepsilon_{2}) (\varepsilon_{1}-2 \varepsilon_{2}) (\varepsilon_{1}-\varepsilon_{2})^2 (3
   \varepsilon_{1}-\varepsilon_{2}) (4 \varepsilon_{1}-\varepsilon_{2}) \varepsilon_{2}^3} \\
 Z^{Sp(0)}_{5}\left(\ \ydiagram{1,1,1,1,1}\ \right)&=-\frac{1}{640 \varepsilon_{1} (\varepsilon_{1}-4 \varepsilon_{2}) (\varepsilon_{1}-3 \varepsilon_{2})
   (\varepsilon_{1}-2 \varepsilon_{2}) (\varepsilon_{1}-\varepsilon_{2}) \varepsilon_{2}^5} \\
 Z^{Sp(0)}_{5}\left(\ \ydiagram{1,1,1,1}\ ,\ydiagram{1}\ \right)&=\frac{3 (2 \varepsilon_{1}+5 \varepsilon_{2})}{1280 \varepsilon_{1}^2 (\varepsilon_{1}-3 \varepsilon_{2}) (2
   \varepsilon_{1}-3 \varepsilon_{2}) (\varepsilon_{1}-2 \varepsilon_{2}) (\varepsilon_{1}-\varepsilon_{2})
   \varepsilon_{2}^5} \\
 Z^{Sp(0)}_{5}\left(\ \ydiagram{1,1,1}\ ,\ydiagram{1,1}\ \right)&=-\frac{(2 \varepsilon_{1}+3 \varepsilon_{2}) (2 \varepsilon_{1}+5 \varepsilon_{2})}{960 \varepsilon_{1}^2 (2
   \varepsilon_{1}-3 \varepsilon_{2}) (\varepsilon_{1}-2 \varepsilon_{2}) (\varepsilon_{1}-\varepsilon_{2})^2 (2
   \varepsilon_{1}-\varepsilon_{2}) \varepsilon_{2}^5}
\end{align}
\subsection{$k=6$}\small
\begin{align}
Z^{Sp(0)}_{6}\left(\ \ydiagram{1,1,1,1,1,1}\ \right)&=\frac{1}{4608 \varepsilon_{1} (\varepsilon_{1}-5 \varepsilon_{2}) (\varepsilon_{1}-4 \varepsilon_{2})
   (\varepsilon_{1}-3 \varepsilon_{2}) (\varepsilon_{1}-2 \varepsilon_{2}) (\varepsilon_{1}-\varepsilon_{2})
   \varepsilon_{2}^6}\\
Z^{Sp(0)}_{6}\left(\ \ydiagram[*(gray)]{0,1,1}*[*(white)]{1,1,1,1}\ \right)&=\frac{(\varepsilon_{1}+2 \varepsilon_{2}) (\varepsilon_{1}+3 \varepsilon_{2})}{9216 \varepsilon_{1}^2
   (\varepsilon_{1}-3 \varepsilon_{2}) (\varepsilon_{1}-2 \varepsilon_{2})^2 (\varepsilon_{1}-\varepsilon_{2})^3
   \varepsilon_{2}^6}\\
Z^{Sp(0)}_{6}\left(\ \ydiagram{1+2,1+1,1+1,2}\ \right)&=\frac{5 (2 \varepsilon_{1}+\varepsilon_{2})}{6912 \varepsilon_{1} (\varepsilon_{1}-4 \varepsilon_{2})
   (\varepsilon_{1}-3 \varepsilon_{2}) (2 \varepsilon_{1}-3 \varepsilon_{2}) (\varepsilon_{1}-2 \varepsilon_{2})^2
   (\varepsilon_{1}-\varepsilon_{2})^3 \varepsilon_{2}^4}\\
Z^{Sp(0)}_{6}\left(\ \ydiagram{1,1,1,1}\ ,\ydiagram{2}\ \right)&=-\frac{(\varepsilon_{1}+3 \varepsilon_{2}) (3 \varepsilon_{1}+5 \varepsilon_{2})}{3072 \varepsilon_{1}^3
   (\varepsilon_{1}-5 \varepsilon_{2}) (\varepsilon_{1}-3 \varepsilon_{2}) (\varepsilon_{1}-2 \varepsilon_{2})
   (\varepsilon_{1}-\varepsilon_{2})^3 \varepsilon_{2}^5}\\
Z^{Sp(0)}_{6}\left(\ \ydiagram{2+1,1+2,2,1}\ \right)&=-\frac{(2 \varepsilon_{1}+\varepsilon_{2}) (\varepsilon_{1}+2 \varepsilon_{2})}{768 \varepsilon_{1}^2
   (\varepsilon_{1}-3 \varepsilon_{2}) (2 \varepsilon_{1}-3 \varepsilon_{2}) (\varepsilon_{1}-2 \varepsilon_{2})^2
   (3 \varepsilon_{1}-2 \varepsilon_{2}) (\varepsilon_{1}-\varepsilon_{2})^3 (2 \varepsilon_{1}-\varepsilon_{2})^2
   \varepsilon_{2}^2}\\
Z^{Sp(0)}_{6}\left(\ \ydiagram{1,1}\ ,\ydiagram{1+1,2,1}\ \right)&=-\frac{(3 \varepsilon_{1}+\varepsilon_{2}) (\varepsilon_{1}+2 \varepsilon_{2}) (\varepsilon_{1}+5
   \varepsilon_{2})}{512 \varepsilon_{1}^3 (\varepsilon_{1}-5 \varepsilon_{2}) (\varepsilon_{1}-3 \varepsilon_{2})
   (\varepsilon_{1}-2 \varepsilon_{2}) (\varepsilon_{1}-\varepsilon_{2})^3 (2 \varepsilon_{1}-\varepsilon_{2}) (3
   \varepsilon_{1}-\varepsilon_{2}) \varepsilon_{2}^4}\\
Z^{Sp(0)}_{6}\left(\ \ydiagram{1,1}\ ,\ydiagram{2,2}\ \right)&=\frac{(\varepsilon_{1}+2 \varepsilon_{2}) (3 \varepsilon_{1}+2 \varepsilon_{2}) (3 \varepsilon_{1}+4
   \varepsilon_{2})}{768 \varepsilon_{1}^3 (\varepsilon_{1}-4 \varepsilon_{2}) (\varepsilon_{1}-2 \varepsilon_{2})^2
   (3 \varepsilon_{1}-2 \varepsilon_{2}) (\varepsilon_{1}-\varepsilon_{2})^3 (2 \varepsilon_{1}-\varepsilon_{2})
   \varepsilon_{2}^4}\\
Z^{Sp(0)}_{6}\left(\ \ydiagram{3,3}\ \right)&=\frac{11}{2304 \varepsilon_{1}^3 (\varepsilon_{1}-2 \varepsilon_{2}) (\varepsilon_{1}-\varepsilon_{2})^3 (2
   \varepsilon_{1}-\varepsilon_{2})^2 (3 \varepsilon_{1}-\varepsilon_{2}) \varepsilon_{2}^2}\\
Z^{Sp(0)}_{6}\left(\ \ydiagram{1,1}\ ,\ydiagram[*(gray)]{0,1}*[*(white)]{1,1,1}\ \right)&=\frac{(\varepsilon_{1}+2 \varepsilon_{2}) (2 \varepsilon_{1}+3 \varepsilon_{2})^2 (2 \varepsilon_{1}+5
   \varepsilon_{2})}{23040 \varepsilon_{1}^3 (2 \varepsilon_{1}-3 \varepsilon_{2}) (\varepsilon_{1}-2
   \varepsilon_{2}) (\varepsilon_{1}-\varepsilon_{2})^3 (2 \varepsilon_{1}-\varepsilon_{2})^2 \varepsilon_{2}^6}\\
Z^{Sp(0)}_{6}\left(\ \ydiagram[*(gray)]{1+1,1+1}*[*(white)]{1+2,2}\ \right)&=\frac{(2 \varepsilon_{1}+\varepsilon_{2})^2 (\varepsilon_{1}+2 \varepsilon_{2})}{1536 \varepsilon_{1}^3
   (\varepsilon_{1}-2 \varepsilon_{2})^2 (\varepsilon_{1}-\varepsilon_{2})^3 (2 \varepsilon_{1}-\varepsilon_{2})^2
   (3 \varepsilon_{1}-\varepsilon_{2}) \varepsilon_{2}^4}\\
Z^{Sp(0)}_{6}\left(\ \ydiagram{2+3,3}\ \right)&=\frac{(2 \varepsilon_{1}+\varepsilon_{2}) (3 \varepsilon_{1}+\varepsilon_{2})}{768 \varepsilon_{1}^3 (3
   \varepsilon_{1}-2 \varepsilon_{2}) (\varepsilon_{1}-\varepsilon_{2})^3 (2 \varepsilon_{1}-\varepsilon_{2})^2 (3
   \varepsilon_{1}-\varepsilon_{2}) (4 \varepsilon_{1}-\varepsilon_{2}) (5 \varepsilon_{1}-\varepsilon_{2})
   \varepsilon_{2}^2}\\
Z^{Sp(0)}_{6}\left(\ \ydiagram{2}\ ,\ydiagram{1+2,2}\ \right)&=-\frac{(2 \varepsilon_{1}+\varepsilon_{2}) (5 \varepsilon_{1}+\varepsilon_{2}) (\varepsilon_{1}+3
   \varepsilon_{2})}{512 \varepsilon_{1}^4 (\varepsilon_{1}-3 \varepsilon_{2}) (\varepsilon_{1}-2 \varepsilon_{2})
   (\varepsilon_{1}-\varepsilon_{2})^3 (2 \varepsilon_{1}-\varepsilon_{2}) (3 \varepsilon_{1}-\varepsilon_{2}) (5
   \varepsilon_{1}-\varepsilon_{2}) \varepsilon_{2}^3}\\
Z^{Sp(0)}_{6}\left(\ \ydiagram{1,1}\ ,\ydiagram{4}\ \right)&=-\frac{(3 \varepsilon_{1}+\varepsilon_{2}) (5 \varepsilon_{1}+3 \varepsilon_{2})}{3072 \varepsilon_{1}^5
   (\varepsilon_{1}-\varepsilon_{2})^3 (2 \varepsilon_{1}-\varepsilon_{2}) (3 \varepsilon_{1}-\varepsilon_{2}) (5
   \varepsilon_{1}-\varepsilon_{2}) \varepsilon_{2}^3}\\
Z^{Sp(0)}_{6}\left(\ \ydiagram{1,1}\ ,\ydiagram[*(gray)]{1+1}*[*(white)]{3}\ \right)&=-\frac{(2 \varepsilon_{1}+\varepsilon_{2})^2 (2 \varepsilon_{1}+3 \varepsilon_{2}) (4 \varepsilon_{1}+3
   \varepsilon_{2})}{4608 \varepsilon_{1}^5 (2 \varepsilon_{1}-3 \varepsilon_{2})
   (\varepsilon_{1}-\varepsilon_{2})^3 (2 \varepsilon_{1}-\varepsilon_{2})^2 (4 \varepsilon_{1}-\varepsilon_{2})
   \varepsilon_{2}^4}\\
Z^{Sp(0)}_{6}\left(\ \ydiagram[*(gray)]{0,0,2}*[*(white)]{1+1,1+1,2,1,1}\ \right)&=-\frac{(\varepsilon_{1}+2 \varepsilon_{2}) (\varepsilon_{1}+3 \varepsilon_{2})}{768 \varepsilon_{1}^2
   (\varepsilon_{1}-5 \varepsilon_{2}) (\varepsilon_{1}-4 \varepsilon_{2}) (\varepsilon_{1}-3 \varepsilon_{2}) (2
   \varepsilon_{1}-3 \varepsilon_{2}) (\varepsilon_{1}-2 \varepsilon_{2})^2 (\varepsilon_{1}-\varepsilon_{2})^3
   \varepsilon_{2}^3}\\
Z^{Sp(0)}_{6}\left(\ \ydiagram{1+1,2,2,1}\ \right)&=\frac{\varepsilon_{1}+3 \varepsilon_{2}}{144 \varepsilon_{1}^2 (\varepsilon_{1}-4 \varepsilon_{2})
   (\varepsilon_{1}-3 \varepsilon_{2}) (2 \varepsilon_{1}-3 \varepsilon_{2}) (\varepsilon_{1}-2 \varepsilon_{2})^2
   (\varepsilon_{1}-\varepsilon_{2})^2 (2 \varepsilon_{1}-\varepsilon_{2}) \varepsilon_{2}^3}\\
Z^{Sp(0)}_{6}\left(\ \ydiagram{2,2,2}\ \right)&=-\frac{11}{2304 \varepsilon_{1}^2 (\varepsilon_{1}-3 \varepsilon_{2}) (\varepsilon_{1}-2 \varepsilon_{2})^2
   (\varepsilon_{1}-\varepsilon_{2})^3 (2 \varepsilon_{1}-\varepsilon_{2}) \varepsilon_{2}^3}\\
Z^{Sp(0)}_{6}\left(\ \ydiagram[*(gray)]{0,0,1}*[*(white)]{1,1,1,1,1}\ \right)&=-\frac{\varepsilon_{1}+3 \varepsilon_{2}}{2880 \varepsilon_{1}^2 (\varepsilon_{1}-4 \varepsilon_{2})
   (\varepsilon_{1}-3 \varepsilon_{2}) (\varepsilon_{1}-2 \varepsilon_{2})^2 (\varepsilon_{1}-\varepsilon_{2})
   \varepsilon_{2}^6}\\
Z^{Sp(0)}_{6}\left(\ \ydiagram[*(gray)]{0,2}*[*(white)]{1+1,2,1}\ \right)&=\frac{(2 \varepsilon_{1}+\varepsilon_{2}) (\varepsilon_{1}+2 \varepsilon_{2})^2}{1536 \varepsilon_{1}^4
   (\varepsilon_{1}-3 \varepsilon_{2}) (\varepsilon_{1}-2 \varepsilon_{2})^2 (\varepsilon_{1}-\varepsilon_{2})^3 (2
   \varepsilon_{1}-\varepsilon_{2})^2 \varepsilon_{2}^3}\\
Z^{Sp(0)}_{6}\left(\ \ydiagram{2+2,1+2,2}\ \right)&=-\frac{(2 \varepsilon_{1}+\varepsilon_{2}) (\varepsilon_{1}+2 \varepsilon_{2})}{768 \varepsilon_{1}^2 (2
   \varepsilon_{1}-3 \varepsilon_{2}) (\varepsilon_{1}-2 \varepsilon_{2})^2 (3 \varepsilon_{1}-2 \varepsilon_{2})
   (\varepsilon_{1}-\varepsilon_{2})^3 (2 \varepsilon_{1}-\varepsilon_{2})^2 (3 \varepsilon_{1}-\varepsilon_{2})
   \varepsilon_{2}^2}\\
Z^{Sp(0)}_{6}\left(\ \ydiagram{3+1,4,1}\ \right)&=\frac{5 (\varepsilon_{1}+2 \varepsilon_{2})}{6912 \varepsilon_{1}^4 (3 \varepsilon_{1}-2 \varepsilon_{2})
   (\varepsilon_{1}-\varepsilon_{2})^3 (2 \varepsilon_{1}-\varepsilon_{2})^2 (3 \varepsilon_{1}-\varepsilon_{2}) (4
   \varepsilon_{1}-\varepsilon_{2}) \varepsilon_{2}}\\
Z^{Sp(0)}_{6}\left(\ \ydiagram{1+3,3}\ \right)&=\frac{3 \varepsilon_{1}+\varepsilon_{2}}{144 \varepsilon_{1}^3 (\varepsilon_{1}-2 \varepsilon_{2}) (3
   \varepsilon_{1}-2 \varepsilon_{2}) (\varepsilon_{1}-\varepsilon_{2})^2 (2 \varepsilon_{1}-\varepsilon_{2})^2 (3
   \varepsilon_{1}-\varepsilon_{2}) (4 \varepsilon_{1}-\varepsilon_{2}) \varepsilon_{2}^2}\\
Z^{Sp(0)}_{6}\left(\ \ydiagram{2}\ ,\ydiagram{2,2}\ \right)&=\frac{(2 \varepsilon_{1}+\varepsilon_{2}) (2 \varepsilon_{1}+3 \varepsilon_{2}) (4 \varepsilon_{1}+3
   \varepsilon_{2})}{768 \varepsilon_{1}^4 (2 \varepsilon_{1}-3 \varepsilon_{2}) (\varepsilon_{1}-2 \varepsilon_{2})
   (\varepsilon_{1}-\varepsilon_{2})^3 (2 \varepsilon_{1}-\varepsilon_{2})^2 (4 \varepsilon_{1}-\varepsilon_{2})
   \varepsilon_{2}^3}\\
Z^{Sp(0)}_{6}\left(\ \ydiagram[*(gray)]{0,1+1}*[*(white)]{1+1,3,1+1}\ \right)&=\frac{(\varepsilon_{1}+\varepsilon_{2})^2}{432 \varepsilon_{1}^4 (\varepsilon_{1}-2 \varepsilon_{2})^2
   (\varepsilon_{1}-\varepsilon_{2})^2 (2 \varepsilon_{1}-\varepsilon_{2})^2 \varepsilon_{2}^4}\\
Z^{Sp(0)}_{6}\left(\ \ydiagram[*(gray)]{0,1+1}*[*(white)]{1+2,1+1,2}\ \right)&=-\frac{(2 \varepsilon_{1}+\varepsilon_{2}) (\varepsilon_{1}+2 \varepsilon_{2})}{96 \varepsilon_{1}^2
   (\varepsilon_{1}-3 \varepsilon_{2}) (2 \varepsilon_{1}-3 \varepsilon_{2}) (\varepsilon_{1}-2 \varepsilon_{2})^2
   (3 \varepsilon_{1}-2 \varepsilon_{2}) (\varepsilon_{1}-\varepsilon_{2}) (2 \varepsilon_{1}-\varepsilon_{2})^2
   \varepsilon_{2}^4}\\
Z^{Sp(0)}_{6}\left(\ \ydiagram{2}\ ,\ydiagram[*(gray)]{0,1}*[*(white)]{1,1,1}\ \right)&=\frac{(\varepsilon_{1}+2 \varepsilon_{2})^2 (3 \varepsilon_{1}+2 \varepsilon_{2}) (3 \varepsilon_{1}+4
   \varepsilon_{2})}{4608 \varepsilon_{1}^4 (\varepsilon_{1}-4 \varepsilon_{2}) (\varepsilon_{1}-2
   \varepsilon_{2})^2 (3 \varepsilon_{1}-2 \varepsilon_{2}) (\varepsilon_{1}-\varepsilon_{2})^3 \varepsilon_{2}^5}\\
Z^{Sp(0)}_{6}\left(\ \ydiagram[*(gray)]{0,1+1}*[*(white)]{2+1,3,1}\ \right)&=-\frac{(2 \varepsilon_{1}+\varepsilon_{2}) (\varepsilon_{1}+2 \varepsilon_{2})}{96 \varepsilon_{1}^4 (2
   \varepsilon_{1}-3 \varepsilon_{2}) (\varepsilon_{1}-2 \varepsilon_{2})^2 (3 \varepsilon_{1}-2 \varepsilon_{2})
   (\varepsilon_{1}-\varepsilon_{2}) (2 \varepsilon_{1}-\varepsilon_{2})^2 (3 \varepsilon_{1}-\varepsilon_{2})
   \varepsilon_{2}^2}\\
Z^{Sp(0)}_{6}\left(\ \ydiagram{6}\ \right)&=-\frac{1}{4608 \varepsilon_{1}^6 (\varepsilon_{1}-\varepsilon_{2}) (2 \varepsilon_{1}-\varepsilon_{2}) (3
   \varepsilon_{1}-\varepsilon_{2}) (4 \varepsilon_{1}-\varepsilon_{2}) (5 \varepsilon_{1}-\varepsilon_{2})
   \varepsilon_{2}}\\
Z^{Sp(0)}_{6}\left(\ \ydiagram[*(gray)]{1+2}*[*(white)]{4}\ \right)&=\frac{(2 \varepsilon_{1}+\varepsilon_{2}) (3 \varepsilon_{1}+\varepsilon_{2})}{9216 \varepsilon_{1}^6
   (\varepsilon_{1}-\varepsilon_{2})^3 (2 \varepsilon_{1}-\varepsilon_{2})^2 (3 \varepsilon_{1}-\varepsilon_{2})
   \varepsilon_{2}^2}\\
Z^{Sp(0)}_{6}\left(\ \ydiagram{2}\ ,\ydiagram[*(gray)]{1+1}*[*(white)]{3}\ \right)&=-\frac{(2 \varepsilon_{1}+\varepsilon_{2}) (3 \varepsilon_{1}+2 \varepsilon_{2})^2 (5 \varepsilon_{1}+2
   \varepsilon_{2})}{23040 \varepsilon_{1}^6 (\varepsilon_{1}-2 \varepsilon_{2})^2 (3 \varepsilon_{1}-2
   \varepsilon_{2}) (\varepsilon_{1}-\varepsilon_{2})^3 (2 \varepsilon_{1}-\varepsilon_{2}) \varepsilon_{2}^3}\\
Z^{Sp(0)}_{6}\left(\ \ydiagram[*(gray)]{2+1}*[*(white)]{5}\ \right)&=\frac{3 \varepsilon_{1}+\varepsilon_{2}}{2880 \varepsilon_{1}^6 (\varepsilon_{1}-\varepsilon_{2}) (2
   \varepsilon_{1}-\varepsilon_{2})^2 (3 \varepsilon_{1}-\varepsilon_{2}) (4 \varepsilon_{1}-\varepsilon_{2})
   \varepsilon_{2}^2}
\end{align}
\large
\subsection{$k=7$}\scriptsize

\large
\subsection{$k=8$}\tiny
%
\normalsize
Note that the numerator in (\ref{doublepole}) is not factorized. The degree $-(4+1)$ part appears in its formal series around the pole (\ref{doublepole}) and then one should pick up the first degree of the Taylor series of the nonzero factors.

\section{A Comment on the multiplicity factors}\label{multiplicityfactors}
To see relations between the $A_{\sigma,\eta}$ and the notion of the multiplicity factor in \cite{Hollands:2010xa}, we rephrase the distinction rule as follows.
\begin{proposition}
Given a graph of the set $\{Q_{1},\cdots,Q_{n}\}\subset\widetilde{{\bf Q}}^{BCD}$ which is of the form (\ref{loop}), the distinction rule for the graph is equivalent to the following condition; for any $1\le i\le n$, there is a number $j\ge i$ and at least one of two vector $e_{i}\pm e_{j}$ ($e_{i}$ for $j=i$) can be expressed as a non-negative linear combinations of $Q_{1},\cdots,Q_{n}$.\hspace{\fill}$\blacksquare$
\end{proposition}
\begin{proof}
 Firstly, we show that the latter condition is satisfied when the graph is $\eta$-oriented. We use the notation in the section \ref{section3}.

We first consider a vertex $v\in V_{1}\sqcup V_{2}$. For the case $(Q_{v}, e_{m(v)})=-1$, the sum of the lines in the path between $v$ and the vertex $w\neq v$ with $m(w)=M(v)$, becomes a desired vector for the number $m(v)$.

For the other case $(Q_{v}, e_{m(v)})=+1$, we trace back to the ascendants of $v$. Let $u$ be the parent vertex of  $v$. If $m(v)<m(u)$, $Q_{v}$ is a vector that we want to have. If not, $u\in V_{2}\sqcup V_{3}\sqcup V_{4}\sqcup V_{5}$. Here we assume that $u\in V_{2}$. Thanks to $\pm p_{u}$ in (\ref{ass2}), there is a line $Q\neq Q_{v}$ that satisfies $(Q+Q_{v}, e_{m(u)})=0$. Here we can choose such $Q$ from $Q_{u}$ and the vector $Q_{T}$, the line between $u$ and its descendant tree $T$ in (\ref{ass2}). If we can choose $Q=Q_{T}$, we see that there exists a vertex $w$ in $T$, with $m(w)=M(u)>m(v)$ and then the sum of all the lines in the unique path between $v$ and $w$ becomes a desired vector for the $m(v)$. If we can not choose $Q=Q_{T}$, we choose $Q=Q_{u}$ and repeat the above process. 

As a result, we can get a desired vector or a vector $e_{m(v)}\pm e_{m(u)}$ by a non-negative combination of $Q$'s, where $u=u(v)$ is the unique vertex satisfying $u\in V_{3}\sqcup V_{4}\sqcup V_{5}$ and $v\in D(u)$. We focus on the latter vector. We may assume $m(u)<m(v)$, or $e_{m(v)}\pm e_{m(u)}$ is a desired vector.

If $u\in V_{5}$, an appropriate path around the loop gives a vector $e_{m(v)}+e_{m(t)}$ by a non-negative combination of $Q$'s plus $e_{m(v)}\pm e_{m(u)}$, where $t$ is the unique vertex in the $V_{3}\sqcup V_{4}$. Then we go down in the tree of the $t$ to the vertex $w$ with $m(w)=n$, and then we get a desired vector $e_{m(v)}+e_{n}$. On the other hand, if $u\in V_{3}\sqcup V_{4}$, the vector $e_{m(v)}\pm e_{m(u)}$ plus the half of the sum of all the lines in the loop becomes a desired vector $e_{m(v)}$. So we can get a desired vector for each vertex $v\in V_{1}\sqcup V_{2}$.

Next we consider vertices in the loop. If $v\in V_{3}\sqcup V_{5}$, repeat the above discussion and get a desired vector for the $m(v)$. If $v\in V_{4}$, the sum of all the lines in the path from $v$ to the vertex $w$ with $m(w)=n$ is a desired vector $e_{m(u_{1})}+e_{n}$. 

As a result, we can get a desired vector for each $1\le i\le n$.

Conversely, we assume the latter condition. Take a vertex $v$ that is not in the loop of the graph. The rule (\ref{ass1}) is trivial for the case $v\in V_{1}$. For the other case $v\in V_{2}$, an application of the assumption to the vertex $w\in D(v)$ with $m(w)=M(v)$ leads the rule (\ref{ass2}).

Next, if $v\in V_{3}$, we have to make $e_{n}$ by a non-negative combination of the $Q$s. This leads the rule (\ref{general1}). If $V_{3}$ is empty, we apply the condition to the vertex $w$ with $m(w)=n$, and then we get the rules (\ref{general2},\ref{general3}) for the case $v\in V_{4}\sqcup V_{5}$. As a result, we derive the distinction rule from the latter condition.

Thus, the two conditions are equivalent.
\end{proof}

Now we interpret the distinction rule in the box language. We have the following corollaries that judge whether a graph contains an $\eta$-oriented subgraph or not.
\begin{corollary}
The graph for a pole $\phi_{*}\in \mathbb{C}^{n}$ contains an $\eta$-oriented subgraph if and only if, for each box with number $1\le i\le n$ in the box expression of $\phi_{*}$, there is another box with number $j(>i)$ or character $j' (j\ge i)$ in its upper-left region in the lattice where the box belongs.\hspace{\fill}$\blacksquare$
\end{corollary}
\begin{corollary}
The graph for a pole $\phi_{*}\in \mathbb{C}^{n}$ is $\eta$-oriented if and only if for each box with number $1\le i\le n$ in the box expression of $\phi_{*}$, there is a box with number $j(>i)$ or character $j' (j\ge i)$ in its upper-left region in the lattice where the box belongs.\hspace{\fill}$\blacksquare$
\end{corollary}
The latter corollary provides an algorithm to determine the multiplicity factor that appeared in \cite{Hollands:2010xa}. 
\begin{example}
Consider a Weyl orbit whose diagram is of the form \ytableausetup{centertableaux,boxsize=.6em}\ydiagram{4}\ytableausetup{centertableaux,boxsize=1.5em}. The corresponding graph to the diagram is\!\!\!
\begin{tikzpicture}[auto,node distance=1cm,
  thick,main node/.style={circle,draw,minimum size=3mm},baseline=-.1cm]
  \node[main node] (1) {$$};
  \node[above] at (1.north){$$};
  \node[main node] (2) [right of=1] {$$};
  \node[above] at (2.north){$$};
  \draw (1) to [in=150,out=210,loop] node {$$}  node[near end, above]{$$}(1);
  \draw (1) to node[very near start,below] {$$} node[above]{$$} node[very near end,below]{$$}(2);
  \end{tikzpicture}. For the box \ytableausetup{centertableaux,boxsize=1em}\begin{ytableau}1\end{ytableau}, at least one of the three boxes \begin{ytableau}1'\end{ytableau}, \begin{ytableau}2\end{ytableau} and \begin{ytableau}2'\end{ytableau} is in the left of \begin{ytableau}1\end{ytableau}. Also, \begin{ytableau}2'\end{ytableau} must lie in the left of \begin{ytableau}2\end{ytableau}. So just the three poles \begin{ytableau}2'&1'&1&2\end{ytableau}, \begin{ytableau}2'&1&1'&2\end{ytableau} and \begin{ytableau}1'&2'&2&1\end{ytableau} give $\eta$-oriented graphs. 
  
Similarly, for the diagram \ytableausetup{centertableaux,boxsize=.6em}\ydiagram{1+2,2}\ytableausetup{centertableaux,boxsize=1em}, we have just one pole \begin{ytableau}\none&2'&1\\1'&2\end{ytableau} whose graph is $\eta$-oriented. As a result, we have the multiplicity factors \ytableausetup{centertableaux,boxsize=.4em}$A_{\ydiagram{4},\eta}=3$ and $A_{\ydiagram{1+2,2},\eta}=1$. \hspace{\fill}$\blacksquare$
\end{example}
\bibliographystyle{utphys}
\bibliography{reference}
\end{document}